\documentclass[10pt,journal,cspaper,compsoc]{IEEEtran}

\usepackage{epsfig}
\usepackage{graphics}
\usepackage{graphicx}
\usepackage{grffile}
\usepackage{amssymb}
\usepackage{amsmath}
\usepackage{subfigure}
\usepackage{algorithm}
\usepackage{algpseudocode}
\newtheorem{myDef}{Definition}
\newtheorem{myPro}{Proposition}
\newtheorem{prop}{Proposition}
\newtheorem{corollary}{Corollary}

\begin{document}
\title{Building Confidential and Efficient Query Services in the Cloud with RASP Data Perturbation}

\author{Huiqi Xu, Shumin Guo, Keke Chen\\
Data Intensive Analysis and Computing Lab\\
Ohio Center of Excellence in Knowledge Enabled Computing\\
Department of Computer Science and Engineering\\
Wright State University, Dayton, OH 45435}

\IEEEcompsoctitleabstractindextext{%
\begin{abstract}
With the wide deployment of public cloud computing infrastructures, using clouds to host data query services has become an appealing solution for the advantages on scalability and cost-saving. However, some data might be sensitive that the data owner does not want to move to the cloud unless the data confidentiality and query privacy are guaranteed. On the other hand, a secured query service should still provide efficient query processing and significantly reduce the in-house workload to fully realize the benefits of cloud computing.  We propose the RASP data perturbation method to provide secure and efficient range query and kNN query services for protected data in the cloud. The RASP data perturbation method combines order preserving encryption, dimensionality expansion, random noise injection, and random projection, to provide strong resilience to attacks on the perturbed data and queries. It also preserves multidimensional ranges, which allows existing indexing techniques to be applied to speedup range query processing. The kNN-R algorithm is designed to work with the RASP range query algorithm to process the kNN queries. We have carefully analyzed the attacks on data and queries under a precisely defined threat model and realistic security assumptions. Extensive experiments have been conducted to show the advantages of this approach on efficiency and security.
\end{abstract}

}

\maketitle
\IEEEdisplaynotcompsoctitleabstractindextext
\IEEEpeerreviewmaketitle

\section{Introduction}
\label{sec:introduction}

Hosting data-intensive query services in the cloud is increasingly popular because of the unique advantages in scalability and cost-saving. With the cloud infrastructures, the service owners can conveniently scale up or down the service and only pay for the hours of using the servers. This is an attractive feature because the workloads of query services are highly dynamic, and it will be expensive and inefficient to serve such dynamic workloads with in-house infrastructures \cite{cloud09}. However, because the service providers lose the control over the data in the cloud, data confidentiality and query privacy have become the major concerns. Adversaries, such as curious service providers, can possibly make a copy of the database or eavesdrop users' queries, which will be difficult to detect and prevent in the cloud infrastructures.

While new approaches
are needed to preserve data confidentiality and query privacy, the efficiency of query services and the benefits of using the clouds should also be preserved. It will not be meaningful to provide slow query services as a result of security and privacy assurance. It is also not practical for the data owner to use a significant amount of in-house resources, because the purpose of using cloud resources is to reduce the need of maintaining scalable in-house infrastructures. Therefore, there is an intricate relationship among the data confidentiality, query privacy, the quality of service, and the economics of using the cloud.

We summarize these requirements for constructing a practical query service in the cloud as the CPEL criteria: data Confidentiality,  query Privacy, Efficient query processing, and Low in-house processing cost. Satisfying these requirements will dramatically increase the complexity of constructing query services in the cloud. Some related approaches have been developed to address some aspects of the problem. However, they do not satisfactorily address all of these aspects. For example, the crypto-index \cite{hakan02sigmod} and Order Preserving Encryption (OPE) \cite{rakesh04} are vulnerable to the attacks. The enhanced crypto-index approach \cite{hore04} puts heavy burden on the in-house infrastructure to improve the security and privacy. The New Casper approach \cite{mokbel06} uses cloaking boxes to protect data objects and queries, which affects the efficiency of query processing and the in-house workload. We have summarized the weaknesses of the existing approaches in Section \ref{sec:related-work}.

We propose the RAndom Space Perturbation (RASP) approach to constructing practical range query and k-nearest-neighbor (kNN) query services in the cloud. The proposed approach will address all the four aspects of the CPEL criteria and aim to achieve a good balance on them. The basic idea is to randomly transform the multidimensional datasets with a combination of order preserving encryption, dimensionality expansion, random noise injection, and random project, so that the utility for processing range queries is preserved. The RASP perturbation is designed in such a way that the queried ranges are securely transformed into polyhedra in the RASP-perturbed data space, which can be efficiently processed with the support of indexing structures in the perturbed space. The RASP kNN query service (kNN-R) uses the RASP range query service to process kNN queries. The key components in the RASP framework include (1) the definition and properties of RASP perturbation; (2) the construction of the privacy-preserving range query services; (3) the construction of privacy-preserving kNN query services; and (4) an analysis of the attacks on the RASP-protected data and queries.

In summary, the proposed approach has a number of unique contributions.
\begin{itemize}
\item The RASP perturbation is a unique combination of OPE, dimensionality expansion, random noise injection, and random projection, which provides strong confidentiality guarantee. 
\item The RASP approach preserves the topology of multidimensional range in secure transformation, which allows indexing and efficiently query processing.

\item The proposed service constructions are able to minimize the in-house processing workload because of the low perturbation cost and high precision query results. This is an important feature enabling practical cloud-based solutions.
\end{itemize}
We have carefully evaluated our approach with synthetic and real datasets. The results show its unique advantages on all aspects of the CPEL criteria.

The entire paper is organized as follows. In Section \ref{sec:RASP}, we define the RASP perturbation method, describe its major properties, and analyze the attacks to the RASP perturbed data. We also introduce the framework for constructing the query services with the RASP perturbation. In Section \ref{sec:RANGE} we describe the algorithm for transforming queries and processing range queries. In Section \ref{sec:kNN-R}, the range query service is extended to handle kNN queries. When describing these two services, we also analyze the attacks on the query privacy. Finally, we present some related approaches in Section \ref{sec:related-work} and analyze their weaknesses in terms of the CPEL criteria.

\section{Query Services in the Cloud} \label{sec:pre}
This section presents the notations, the system architecture, and the threat model for the RASP approach, and prepares for the security analysis \cite{bau11} in later sections.
The design of the system architecture keeps the cloud economics in mind so that most data storage and computing tasks will be done in the cloud. The threat model makes realistic security assumptions and clearly defines the practical threats that the RASP approach will address.

\subsection{Definitions and Notations}
First, we establish the notations. For simplicity, we consider only single database tables, which can be the result of denormalization from multiple relations. A database table consists of $n$ records and $d$ searchable attributes. We also frequently refer to an attribute as a dimension or a column, which are exchangeable in the paper. Each record can be represented as a vector in the multidimensional space, denoted by low case letters. If a record $x$ is $d$-dimensional, we say $x\in \mathbb{R}^d$, where $\mathbb{R}^d$ means the d-dimensional vector space.  A table is also treated as a $d\times n$ matrix, with records represented as column vectors. We use capital letters to represent a table, and indexed capital letters, e.g., $X_i$, to represent columns. Each column is defined on a numerical domain. Categorical data columns are allows in range query, which are converted to numerical domains as we will describe in Section \ref{sec:RASP}.

Range query is an important type of query for many data analytic tasks from simple aggregation to more sophisticated machine learning tasks. Let  $T$ be a table and $X_i$, $X_j$, and $X_k$ be the real valued attributes in $T$, and $a$ and $b$ be some constants. Take the counting query for example. A typical range query looks like
\begin{quote}
 \emph{select count(*) from T\\
 where $X_i \in [a_i, b_i] $ and $X_j \in (a_j, b_j)$ and $X_k=a_k$,
}\end{quote}
which calculates the number of records in the range defined by conditions on $X_i$, $X_j$, and $X_k$. Range queries may be applied to arbitrary number of attributes and conditions on these attributes combined with conditional operators ``and''/``or''. We call each part of the query condition that involves only one attribute as a \emph{simple condition}. A simple condition like $X_i \in [a_i, b_i] $ can be described with two half space conditions $X_i \leq b_i$ and $-X_i \leq -a_i$. Without loss of generality, we will discuss how to process half space conditions like $X_i \leq b_i$ in this paper. A slight modification will extend the discussed algorithms to handle other conditions like $X_i < b_i$ and $X_i=b_i$.

kNN query is to find the closest $k$ records to the query point, where the Euclidean distance is often used to measure the proximity. It is frequently used in location-based services for searching the objects close to a query point, and also in machine learning algorithms such as hierarchical clustering and kNN classifier. A kNN query consists of the query point and the number of nearest neighbors, $k$.

\subsection{System Architecture}
We assume that a cloud computing infrastructure, such as Amazon EC2, is used to host the query services and large datasets. The purpose of this architecture is to extend the \emph{proprietary database servers} to the public cloud, or use a hybrid private-public cloud to achieve scalability and reduce costs while maintaining confidentiality.

Each record $x$ in the outsourced database contains two parts: the RASP-processed attributes $D'=F(D, K)$ and the encrypted original records, $Z=E(D, K')$, where $K$ and $K'$ are keys for perturbation and encryption, respectively. The RASP-perturbed data $D'$ are for  indexing and query processing. Figure \ref{fig:framework} shows the system architecture for both RASP-based range query service and kNN service.

There are two clearly separated groups: the trusted parties and the untrusted parties. The trusted parties include the data/service owner, the in-house proxy server, and the authorized users who can only submit queries. The data owner exports the perturbed data to the cloud. Meanwhile, the authorized users can submit range queries or kNN queries to learn statistics or find some records.
The untrusted parties include the curious cloud provider who hosts the query services and the protected database.  The RASP-perturbed data will be used to build indices to support query processing.

There are a number of basic procedures in this framework: (1) $F(D)$ is the RASP perturbation that transforms the original data $D$ to the perturbed data $D'$; (2) $Q(q)$ transforms the original query $q$ to the protected form $q'$ that can be processed on the perturbed data; (3) $H(q', D')$ is the query processing algorithm that returns the result $R'$. When the statistics such as SUM or AVG of a specific dimension are needed, RASP can work with partial homomorphic encryption such as Paillier encryption \cite{paillier99} to compute these statistics on the encrypted data, which are then recovered with the procedure $G(R')$.

\begin{figure}
\centering
\includegraphics[width = .7\linewidth]{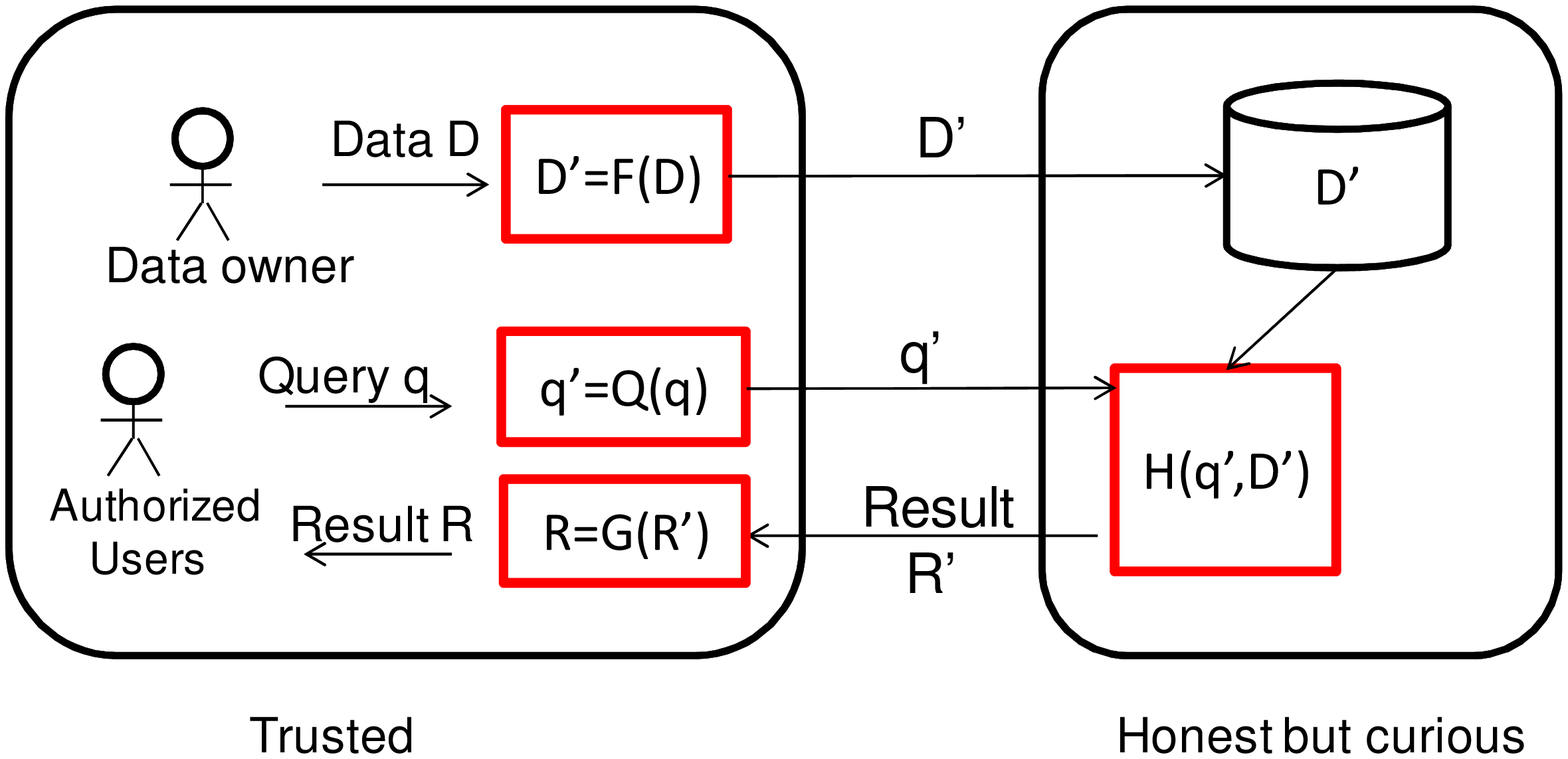}
\caption{The system architecture for RASP-based query services.} \label{fig:framework}
\end{figure}

\subsection{Threat Model}

\textbf{Assumptions.} Our security analysis is built on the important features of the architecture. Under this setting, we believe the following assumptions are appropriate.
\begin{itemize}
\item Only the authorized users can query the proprietary database. Authorized users are not malicious and will not intentionally breach the confidentiality. We consider insider attacks are orthogonal to our research; thus, we can exclude the situation that the authorized users collude with the untrusted cloud providers to leak additional information.
\item The client-side system and the communication channels are properly secured and no protected data records and queries can be leaked.
\item Adversaries can see the perturbed database, the transformed queries, the whole query processing procedure, the access patterns, and understand the same query returns the same set of results, but nothing else.
\item Adversaries can possibly have the global information of the database, such as the applications of the database, the attribute domains, and possibly the attribute distributions, via other published sources (e.g., the distribution of sales, or patient diseases, in public reports).
\end{itemize}

These assumptions can be maintained and reinforced by applying appropriate security policies. Note that this model is equivalent to the eavesdropping model equipped with the plaintext distributional knowledge in the cryptographic setting.

\textbf{Protected Assets.} Data confidentiality and query privacy should be protected in the RASP approach. While the integrity of query services is also an important issue, it is orthogonal to our study. Existing integrity checking and preventing techniques \cite{xie07,sion05,liff06} can be integrated into our framework. Thus, the integrity problem will be excluded from the paper, and we can assume the curious cloud provider is interested in the data and queries, but it will honestly follow the protocol to provide the infrastructure service.

\textbf{Attacker Modeling.} The goal of attack is to recover (or estimate) the original data from the perturbed data, or identify the exact queries (i.e., location queries) to breach users' privacy. According to the level of prior knowledge the attacker may have, we categorize the attacks into two categories.
\begin{itemize}
\item Level 1: The attacker knows only the perturbed data and transformed queries, without any other prior knowledge. This corresponds to the cipertext-only attack in the cryptographic setting.
\item Level 2: The attacker also knows the original data distributions, including individual attribute distributions and the joint distribution (e.g., the covariance matrix) between attributes. In practice, for some applications, whose statistics are interesting to the public domain, the dimensional distributions might have been published via other sources.
\end{itemize}
These levels of knowledge are appropriate according to the assumptions we hold. We will analyze the security based on this threat model.

\textbf{Security Definition.} Different from the traditional encryption schemes, attackers can also be satisfied with good estimation. Therefore, we will investigate two levels of security definitions: (1) it is computationally intractable for the attacker to recover the \emph{exact} original data based on the perturbed data; (2) the attacker cannot \emph{effectively estimate} the original data. The effectiveness measure is defined with the NR\_MSE measure in Section \ref{sec:DataAttack}.


\section{RASP: Random Space Perturbation}\label{sec:RASP}
In this section, we present the basic definition of RAndom Space Perturbation (RASP) method and its properties. We will also discuss the attacks on RASP perturbed data, based on the threat model given in Section \ref{sec:pre}.

\subsection{Definition of RASP}
RASP is one type of multiplicative perturbation, with a novel combination of OPE, dimension expansion, random noise injection, and random projection. Let's consider the multidimensional data are numeric and in multidimensional vector space\footnote{For categorical attributes, we use the following simple mapping because it will not break the query semantics. For a categorical attribute $X_i$, the values $\{c_1,\ldots,c_m\}$ in the domain are mapped to $\{1,\ldots,m\}$. A query condition on categorical values, say $X_i=c_j$, is then converted to $j-\delta\leq X_i\leq j+\delta$, where $\delta \in (0,1)$}. The database has $k$ searchable dimensions and $n$ records, which makes a $d \times n$ matrix $X$. The \emph{searchable} dimensions can be used in queries and thus should be indexed. Let $x$ represent a $d$-dimensional record, $x\in \mathbb{R}^d$. Note that in the $d$-dimensional vector space $\mathbb{R}^d$, the range query conditions are represented as half-space functions  and a range query is translated to finding the point set in corresponding polyhedron area described by the half spaces \cite{boyd04}.

The RASP perturbation involves three steps. Its security is based on the existence of random invertible real-value matrix generator and random real value generator. For each $k$-dimensional input vector $x$,
\begin{enumerate}
\item An order preserving encryption (OPE) scheme \cite{rakesh04}, $E_{ope}$ with keys $K_{ope}$, is applied to each dimension of $x$: $E_{ope}(x, K_{ope}) \in \mathbb{R}^d$ to change the dimensional distributions to normal distributions with each dimension's value order still preserved.
\item The vector is then extended to $d+2$ dimensions as $G(x) = ( (E_{opt}(x))^T, 1, v)^T$, where the $(d+1)$-th dimension is always a $1$ and the $(d+2)$-th dimension, $v$, is drawn from  a random real number generator $RNG$ that generates random values from a tailored normal distributions. We will discuss the design of RNG and OPE later.
\item The $(d+2)$-dimensional vector is finally transformed to
\begin{equation}\label{eq:trans}
F(\mathbf{x}, K=\{A, K_{ope}, RG\}) =  A ( (E_{ope}(x))^T, 1, v)^T ,
\end{equation}
where $A$ is a $(d+2)\times (d+2)$ randomly generated invertible matrix with $a_{ij} \in \mathbb{R}$ such that there are at least two non-zero values in each row of $A$ and the last column of $A$ is also non-zero\footnote{Currently, we use a random invertible matrix generator that draws matrix elements uniformly at random from the standard normal distribution and check the matrix invertibility and the non-zero conditions.}.
\end{enumerate}
$K_{ope}$ and $A$ are shared by all vectors in the database, but $v$ is randomly generated for each individual vector. Since the RASP-perturbed data records are only used for indexing and helping query processing, there is no need to recover the perturbed data. As we mentioned, in the case that original records are needed, the encrypted records associated with the RASP-perturbed records will be returned. We give the detailed algorithm in Appendix.

\textbf{Design of OPE and RNG.}
We use the OPE scheme to convert all dimensions of the original data to the standard normal distribution $\mathcal{N}(0,1)$ in the limited domain $[-\beta,\beta]$. $\beta$ can be selected as a value $>=4$, as the range $[-4, 4]$ covers more than 99\%  of the population.  This can be done with an algorithm such as the one described in \cite{rakesh04}. The use of OPE allows queries to be correctly transformed and processed. Similarly, we draw random noises $v$ from $\mathcal{N}(0,1)$ in the limited domain $[-\beta,\beta]$. Such a design makes the extended noise dimension indifferent from the data dimensions in terms of the distributions.

The design of such an extended data vector $(E_{ope}(x)^T,1,v)^T$ is to enhance the data and query confidentiality. The use of OPE is to transform large-scale or infinite domains to normal distributions, which address the distributional attack. The $(d+1)$-th homogeneous  dimension is for hiding the query content. The $(d+2)$-th dimension injects random noise in the perturbed data and also protects the transformed queries from attacks. The rationale behind different aspects will be discussed clearly in later sections.

\subsection{Properties of RASP}
RASP has several important features. First, RASP does not preserve the order of dimensional values because of the matrix multiplication component, which distinguishes itself from order preserving encryption (OPE) schemes, and thus does not suffer from the distribution-based attack (details in Section \ref{sec:related-work}). An OPE scheme maps a set of single-dimensional values to another, while keeping the value order unchanged. Since the RASP perturbation can be treated as a combined transformation $F(G(E_{ope}(x)))$, it is sufficient to show that $F(y)=Ay$ does not preserve the order of dimensional values, where $y\in \mathbb{R}^{d+2}$ and $A \in \mathbb{R}^{(d+2)\times(d+2)}$. The proof is straightforward as shown in Appendix.

Second, RASP does not preserve the distances between records, which prevents the perturbed data from distance-based attacks \cite{keke07sdm}. Because none of the transformations in the RASP: $E_{ope}$, $G$, and $F$ preserves distances, apparently the RASP perturbation will not preserve distances. Similarly, RASP does not preserve other more sophisticated structures such as covariance matrix and principal components \cite{jolliffe86}. Therefore, the PCA-based attacks such as \cite{huang05,liu06pkdd} do not work as well.

Third, the original range queries can be transformed to the RASP perturbed data space, which is the basis of our query processing strategy. A range query describes a hyper-cubic area (with possibly open bounds) in the multidimensional space. In Section \ref{sec:RANGE}, we will show that a hyper-cubic area in the original space is transformed to a polyhedron with the RASP perturbation. Thus, we can search the points in the polyhedron to get the query results.


\subsection{Data Confidentiality Analysis} \label{sec:DataAttack}
As the threat model describes, attackers might be interested in finding the exact original data records or estimating them based on the perturbed data. For estimation attack, if the estimation is sufficiently accurate (above certain accuracy threshold), we say the perturbation is not secure. Below, we define the measure for evaluating the effectiveness of estimation attacks.

\subsubsection{Evaluating Effectiveness of Estimation Attacks}
Because attackers may not need to exactly recover the original values, an accurate estimation will be sufficient. A measure is needed to define the ``accuracy'' or ``uncertainty'' as we mentioned. We use the commonly used mean-squared-error (MSE) to evaluate the effectiveness of attack. To be semantically consistent, the $j$-th dimension can be treated as sample values drawn from a random variable $X_j$. Let $x_{ij}$ be the value of the $i$-th original record in $j$-th dimension and $\hat{x}_{ij}$ be the estimated value. The MSE for the $j$-th dimension can be defined as
\[
MSE(X_j, \hat{X}_j) = \frac{1}{n}\sum_{i=1}^n (x_{ij}-\hat{x}_{ij})^2,
\]
which is equivalent to the variance: var$(X_j-\hat{X}_j)$. The square root of MSE (RMSE) represent the uncertainty of the estimation - for an estimated value $\hat{x}$, the original value $x$ could be in the range ($\hat{x}$ - RMSE, $\hat{x}$ +RMSE). Thus, the length of the range, 2*RMSE, also represents the accuracy of the estimation.

However, this length is subject to the length of the domain.
Thus, we use the normalized square root of MSE (NR\_MSE).
\begin{equation}
\text{NR\_MSE} (X_j) = 2\sqrt{MSE(X_i, \hat{X_j})}/\text{domain length},
\end{equation}
instead, which is intuitively the rate between the uncertain range and the whole domain.

To compare MSE for multiple columns, we also need to normalize these two series $\{x_{ij}\}$ and $\{\hat{x}_{ij}\}$ to eliminate the difference on domain scales. The normalization procedure \cite{draper98} is described as follows. Assume the mean and variance of the series $\{x_{ij}\}$ is $\mu_j$ and $\sigma_j^2$, correspondingly. The series is transformed by $x_{ij} \leftarrow (x_{ij}-\mu_j)/\sigma_j$. A similar procedure is also applied to the series $\{\hat{x}_{ij}\}$. For the normalized domains, the range $[-2, 2]$ almost covers the whole population\footnote{For a normal distribution $N(\mu, \sigma^2)$, the range $(\mu-2\sigma, \mu+2\sigma)$ covers about 95\% of the population. We use this length $4\sigma$ to approximately represent the majority of population for all other distributions, as normal distribution is a good approximation for many applications.} \cite{draper98}. Therefore, for normalized series, NR\_MSE is simply RMSE/2.

For an attack that can only result in low-accuracy estimation (e.g., NR\_MSE $\geq 20\%$, the uncertainty is more than 20 \% of the domain length.), we call the RASP-perturbed dataset is \emph{resilient} to that attack. Intuitively, NR\_MSE higher than 100\% will not be very meaningful. Thus, we set the absolute upper bound to be 100\%. We will discuss the specific upper bounds according to the level of prior knowledge.

\subsubsection{Prior-Knowledge Based Analysis}
Below, we analyze the security under the two levels of knowledge the attacker may have, according to the two levels of security definitions: exact match and statistical estimation.

\noindent\textbf{Naive Estimation. } We assume each value in the vector or matrix is encoded with $n$ bits.  Let the perturbed vector $p$ be drawn from a random variable $\mathcal{P}$, and the original vector $x$ be drawn from a random variable $\mathcal{X}$. We show that naive estimation is computationally intractable to identify the exact original data with the perturbed data, if we use a random invertible real matrix generator and a random real value generator. The goal is to show  the number of valid $X$ dataset in terms of a known perturbed dataset $P$. Below we discuss a simplified version that contains no OPE component - the OPE version has at least the same level of security. 
\begin{prop}
For a known perturbed dataset $P$, there exists $O(2^{(d+1)(d+2)n})$ candidate $X$ datasets in the original space.
\end{prop}
\begin{proof}
For a given perturbation $P=AZ$, where $Z$ is $X$ with the two extended dimensions,  we use $B_{d+1}$ to represent the $(d+1)$-th row of $A^{-1}$. Thus, $B_{d+1}P = [1,\ldots, 1]$, i.e., the appended $(d+1)$-th row of $Z$. Keeping $B_{d+1}$ unchanged, we randomly generate other rows of $B$ for a candidate $\hat{B}$. The result $\hat{Z}=\hat{B}P$ is a validate estimate of $Z$ if $\hat{B}$ is invertible. Thus, the number of candidate $X$ is the number of invertible $\hat{B}$.

The total number of $\hat{B}$ including non-invertible ones is $2^{(d+1)(d+2)n}$. Based on the theory of invertible random matrix \cite{rudelson09}, the probability of generating a non-invertible random matrix is less than $\exp^{-c(d+2)}$ for some constant $c$. Thus, there are about $(1-\exp^{-c(d+2)})2^{(d+1)(d+2)n}$ invertible $\hat{B}$. Correspondingly, there are a same number of candidate $X$.
\end{proof}
Thus, finding the exact $X$ has a negligible probability in terms of the number of bits, $n$.

As the candidates have an equal probability over the whole domain, according to the definition of NR\_MSE, the uncertain range is the same as the whole domain, resulting in NR\_MSE $=100\%$.

\noindent \textbf{Distribution-based Estimation. }
With the known distributional information, the attacker can do more on estimating the original data. The known most relevant method is called Independent Component Analysis (ICA) \cite{hyvarinen01}. For a multiplicative perturbation $P=AX$, the basic idea is to find an optimal projection, $wP$, where $w$ is a $d+2$ dimension row vector, to result in a row vector with its value distribution close to that of one original attribute. It can be extended to find a matrix $W$, so that $WP$ gives \emph{independent} and \emph{non-gaussian} rows, i.e., a good estimate of $X$.

The ICA algorithms \cite{hyvarinen01,hastie01} are optimization algorithms that try to find such projections by  maximizing the \emph{non-gaussianity}\footnote{Non-gaussianity means the distribution is not normal distribution.} of the projection $wP$. The non-gaussianity of the original attributions is crucial because any projection of a multidimensional normal distribution is still a normal distribution, which leaves no clue for recovery.

Therefore, with our design of OPE and the noise dimension in Section \ref{sec:RASP}, we have the following result.
\begin{prop}
There are $O(2^{dn})$ candidate projection vectors, $w$, that lead to the same level of non-gaussianity.
\end{prop}
\begin{proof}
The OPE encrypted matrix $\bar{X}$ (with the homogeneous dimension excluded, which can be possibly recovered) can be treated as a sample set drawn from a multivariate normal distribution $\mathcal{N}(\mu, \Sigma)$. Any invertible transformation $\bar{P}=\bar{A}\bar{X}$ will result in another multivariate normal distribution $\mathcal{N}(\bar{A}\mu, \bar{A}\Sigma \bar{A}^T)$. Thus, any projection $w\bar{P}$ will not change the gaussianity, and there are $O(2^{dn})$ such candidates of $w$.
\end{proof}
Thus, the probability to identify the right projection is negligible in terms of the number of bits $n$. This shows that any ICA-style estimation that depends on non-guassianity is equally ineffective to the RASP perturbation.

In addition to ICA, Principal Component Analysis (PCA) based attack is another possible distributional attack, which, however, depends on the preservation of covariance matrix \cite{liu06pkdd}. Because the covariance matrix is not preserved in RASP perturbation, the PCA attack cannot be used on RASP perturbed data. It is unknown whether there are other distributional methods for approximately separating $X$ or $A$ from the perturbed data $P$, which will be studied in the ongoing work.

In the worst-case estimation, the attacker can simply draw a sample of $\hat{X}_j$ from the known distribution of the original $X_j$; thus, $X_j$ and $\hat{X}_j$ are independent but have the same distribution. It follows that $MSE=var(X_j-\hat{X}_j) = var(X_j) +var(\hat{X}_j) = 2var(X_j)=2\sigma^2$. Correspondingly, NR\_MSE = $(2\sqrt{MSE})/(4\sigma) =\sqrt{2}/2 \approx 71\%$.

\section{RASP Range-Query Processing} \label{sec:RANGE}
Based on the RASP perturbation method, we design the services for two types of queries: range query and kNN query. This section will dedicate to range query processing. We will first show that a range query in the original space can be transformed to a polyhedron query in the perturbed space, and then we develop a secure way to do the query transformation. Then, we will develop a two-stage query processing strategy for efficient range query processing.

\subsection{Transforming Range Queries}
Let's look at the general form of a range query condition. Let $X_i$ be an attribute in the database. A simple condition in a range query involves only one attribute and is of the form ``\emph{$X_i$ $<$op$>$ $a_i$}'', where $a_i$ is a constant in the normalized domain of $X_i$ and $op\in \{ <,>,=,\leq,\geq,\neq\}$ is a comparison operator. For convenience we will only discuss how to process $X_i< a_i$, while the proposed method can be slightly changed for other conditions. Any complicated range query can be transformed into the disjunction of a set of conjunctions, i.e., $\bigcup_{j=1}^n(\bigcap_{i=1}^mC_{i,j})$, where $m,n$ are some integers depending on the original query conditions and $C_{i,j}$ is a simple condition about $X_i$. Again, to simplify the presentation we restrict our discussion to a single conjunction condition $\cap_{i=1}^mC_{i}$, where $C_i$ is in form of $b_i\leq X_i\leq a_i$. Such a conjunction conditions describes a hyper-cubic area in the multidimensional space.

According to the three nested transformations in RASP $F(G(E_{ope}(x)))$, we will first show that an OPE will transform the original hyper-cubic area to another hyper-cubic area in the OPE space.
\begin{myPro}
Order preserving encryption functions transform a hyper-cubic query range to another hyper-cubic query range.
\end{myPro}
\begin{proof}
The original range query condition consists of simple conditions like $b_i\leq X_i \leq a_i$ for each dimension. Since the order is preserved, each simple condition is transformed as follows: $E_{ope}(b_i)\leq E_{ope}(X_i)\leq E_{ope}(a_i)$, which means the transformed range is still a hyper-cubic query range.
\end{proof}

\begin{figure}[tbh]
\centering
\begin{minipage}{0.45\linewidth}
\centering
\includegraphics[width=\linewidth]{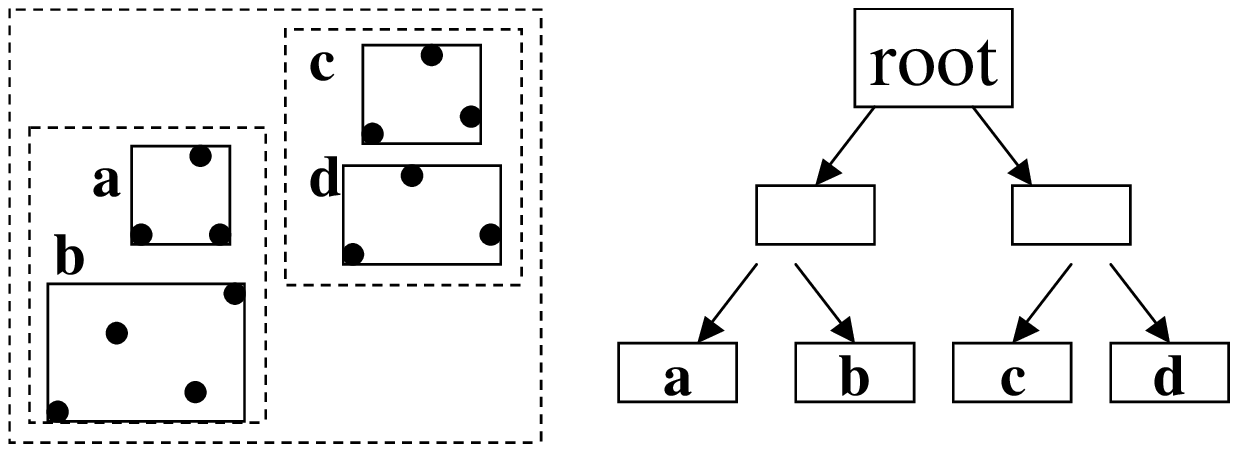}
\caption{R-tree index.}\label{fig:rtree}
\end{minipage}
\begin{minipage}{0.45\linewidth}
\centering
\includegraphics[width=\linewidth]{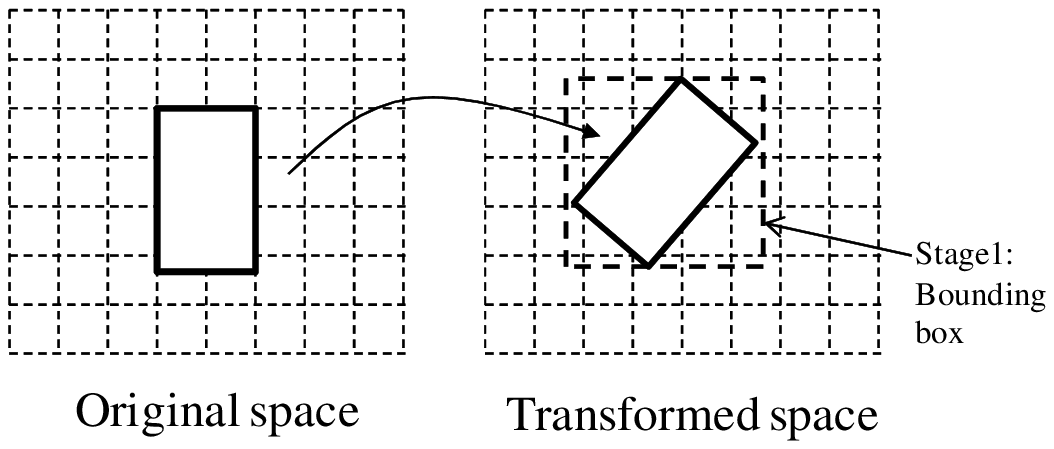}
\caption{Illustration of the two-stage processing algorithm.}\label{fig:two-stage}
\end{minipage}
\end{figure}

Let $y=E_{ope}(x)$ and $c_i =E_{ope}(a_i)$. A simple condition $Y_i\leq c_i$ defines a half-space. With the extended dimensions $z^T=(y^T, 1, v)$,  the half-space can be represented as $w^Tz\leq 0$, where $w$ is a $d+2$ dimensional vector with $w_i=1, w_{d+1}=-c_i$, and $w_{j} = 0$ for $j\neq i,d+1$. Finally, let $u = Az$, according to the RASP transformations. With this representation, the original condition is equivalent to
\begin{equation} \label{eq:q1}
w^TA^{-1}u\leq 0
\end{equation}
in the RASP-perturbed space, which is still a half-space condition. However, this half-space condition will not be parallel to the coordinate - these transformed conditions together form a polyhedron (as illustrated in Figure \ref{fig:two-stage}. The query service will need to find the records in the polyhedron area, which is supported by the two-stage processing algorithm.

\subsection{Security Enhancement on Query Transformation}\label{sec:QuerySecurity}
The attacker may also target on the transformed queries. In this section we discuss such attacks and describe the methods countering the attacks. Note that the attack on small ranges will be described in kNN query processing.

\noindent\textbf{Countering Dimensional Selection Attack}
We show that the dimensional selection attack can reveal partial information of the selected data dimensions, if the attacker knows the distribution of the dimension. Assume the query condition is applied to the $i$-th dimension. If the query parameter $w^TA^{-1}$ is directly submitted to the cloud side, the server can apply $w^TA^{-1}$ to each record $u$ in the server, and get $w^TA^{-1}u = E_{ope}(x_i) - E_{ope}(a_i)$, where $x_i$ is the i-th dimension of the corresponding original record $x$. After getting all such values for the dimension $i$, with the known original data distributions, the attacker can apply the bucket-based distributional attack on the OPE encrypted data (see Section \ref{sec:related-work}) to get an accurate estimate.

According to the design of noise, the extended $(d+2)$-th dimension $v$ in the RASP perturbation: $F(x) = A(E_{ope}(x)^T, 1,v)^T$ is always greater than $v_0$, which can be used to construct secure query conditions. Instead of processing a half space condition $E_{ope}(X_i)\leq E_{ope}(a_i)$, we use $(E_{ope}(X_i)-E_{ope}(a_i)) (v-v_0) \leq 0$ instead. These two conditions are equivalent because $v$ always satisfies $v>v_0$. Using the similar transformations, we get $E_{ope}(X_i) -E_{ope}(a_i) = w^TA^{-1}u$ and $v = \mathbf{q}^TA^{-1}u$, where $q_{d+2}=-1$,  $q_{d+1}=v_0$, and $q_j=0$, for $j\neq d$. Thus, we get the transformed quadratic query condition
\begin{equation}\label{eq:qt}
  u^T(A^{-1})^Twq^TA^{-1}u \leq 0.
\end{equation}
Let $\Theta_i = (A^{-1})^Twq^TA^{-1}$. Now $\Theta$ is submitted to the server and the server will use $u^T\Theta_i u\leq 0$  to filter out the results.

We now show that this query transformation is resilient to the dimensional selection attack. Applying $u^T\Theta u$ to each record $u$, we get $(E_{ope}(X_i)-E_{ope}(a_i))(v-v_0)$. Since $v$ is randomly chosen for each record, the value $E_{ope}(X_i)-E_{ope}(a_i)$ is protected by the randomization.
$\Theta_i$ does not reveal the key parameters as well. Let $c_i= E_{ope}(a_i)$ and $a_{i}$ be the $i$-th row of $A^{-1}$. $\Theta_i$ is $(a_i-c_ia_{d+1})^T(v_0a_{d+1}-a_{d+2})$. As all the components: $a_i, c_i, a_{d+1}$, and $a_{d+2}$ are unknown and cannot be further reduced, $\Theta_i$ provide no information to help drive information about $A^{-1}$.

\noindent\textbf{Other Potential Threats.}
Because the query transformation method does not introduce randomness - the same query will always get the same transformation, and thus the confidentiality of access pattern is not preserved. We summarize the leaked information related to access patterns as follows.
\begin{itemize}
\item Attackers know the exact frequency of each transformed query.
\item The set relationships (set intersection, union, difference, etc.) between the query results are revealed as a result of exact range query processing.
\item Some query matrices on the same dimension may have special relationship preserved as shown in Proposition \ref{prop:theta}, which we will discuss later.
\end{itemize}
We admit this is a weakness of the current design. However, according to the threat model, the adversary will not know any of the original data and queries. Thus, by simply observing the query frequency or relationships between queries, one cannot derive useful information. An important future work is to formally define the specific information leakage caused by the leaked query and access patterns, and then precisely analyze the data and query confidentiality affected by this information leakage under different security assumptions.

\subsection{A Two-Stage Query Processing Strategy with Multidimensional Index Tree}
With the transformed queries, the next important task is to process queries efficiently and return precise results to minimize the client-side post-processing effects. A commonly used method is to use multidimensional tree indices to improve the search performance. However, multidimensional tree indices are normally used to process axis-aligned ``bounding boxes''; whereas, the transformed queries are in arbitrary polyhedra, not necessarily aligned to axes. In this section, we propose a two-stage query processing strategy to handle such irregular-shape queries in the perturbed space. 

\noindent\textbf{Multidimensional Index Tree. } Most multidimensional indexing algorithms are derived from R-tree like algorithms \cite{rtreebook}, where the axis-aligned minimum bounding region (MBR) is the construction block for indexing the multidimensional data. For 2D data, an MBR is a rectangle. For higher dimensions, the shape of MBR is extended to hyper-cube. Figure \ref{fig:rtree} shows the MBRs in the R-tree for a 2D dataset, where each node is bounded by a node MBR. The R-tree range query algorithm compares the MBR and the queried range to find the answers.


\noindent\textbf{The Two-Stage Processing Algorithm. }
%
The transformed query describes a polyhedron in the perturbed space that cannot be directly processed by multidimensional tree algorithms. New tree search algorithms could be designed to use arbitrary polyhedron conditions directly for search. However, we use a simpler two-stage solution that keeps the existing tree search algorithms unchanged.

At the first stage, the proxy in the client side finds the MBR of the polyhedron (as a part of the submitted transformed query) and submit the MBR and a set of secured query conditions $\{\Theta_1,\ldots,\Theta_m\}$ to the server. The server then uses the tree index to find the set of records enclosed by the MBR.

The MBR of the polyhedron can be efficiently founded based on the original range. The original query condition constructs a hyper-cube shape. With the described query transformation, the vertices of the hyper cube are also transformed to vertices of the polyhedron. Therefore, the MBR of the vertices is also the MBR of the polyhedron \cite{franco85}.
Figure \ref{fig:two-stage} illustrates the relationship between the vertices and the MBR and the two-stage processing strategy.


At the second stage, the server uses the transformed halfspace conditions to filter the initial result. In most cases of tight ranges, the initial result set will be reasonably small so that it can be filtered in memory by simply checking the transformed half-space conditions. However, in the worst case, the MBR of the polyhedron will possibly enclose the entire dataset and the second stage is reduced to a linear scan of the entire dataset. The result of second stage will return the \emph{exact} range query result to the proxy server, which significantly reduces the post-processing cost that the proxy server needs to take. It is very important to the cloud-based service, because low post-processing cost requires low in-house investment.


\section{KNN Query Processing with RASP}\label{sec:kNN-R}
Because the RASP perturbation does not preserve distances (and distance orders), kNN query cannot be directly processed with the RASP perturbed data. In this section, we design a kNN query processing algorithm based on range queries (the kNN-R algorithm). As a result, the use of index in range query processing also enables fast processing of kNN queries.

\subsection{Overview of the kNN-R Algorithm}
The original distance-based kNN query processing finds the nearest $k$ points in the \emph{spherical range} that is centered at the query point. The basic idea of our algorithm is to use square ranges, instead of spherical ranges, to find the approximate kNN results, so that the RASP range query service can be used. There are a number of key problems to make this work securely and efficiently. (1) How to efficiently find the minimum square range that surely contains the k results, without many interactions between the cloud and the client? (2) Will this solution preserve data confidentiality and query privacy? (3) Will the proxy server's workload increase? to what extent?

\begin{figure*}
\centering
\begin{minipage}{0.45\linewidth}
\centering
\includegraphics[width =\linewidth]{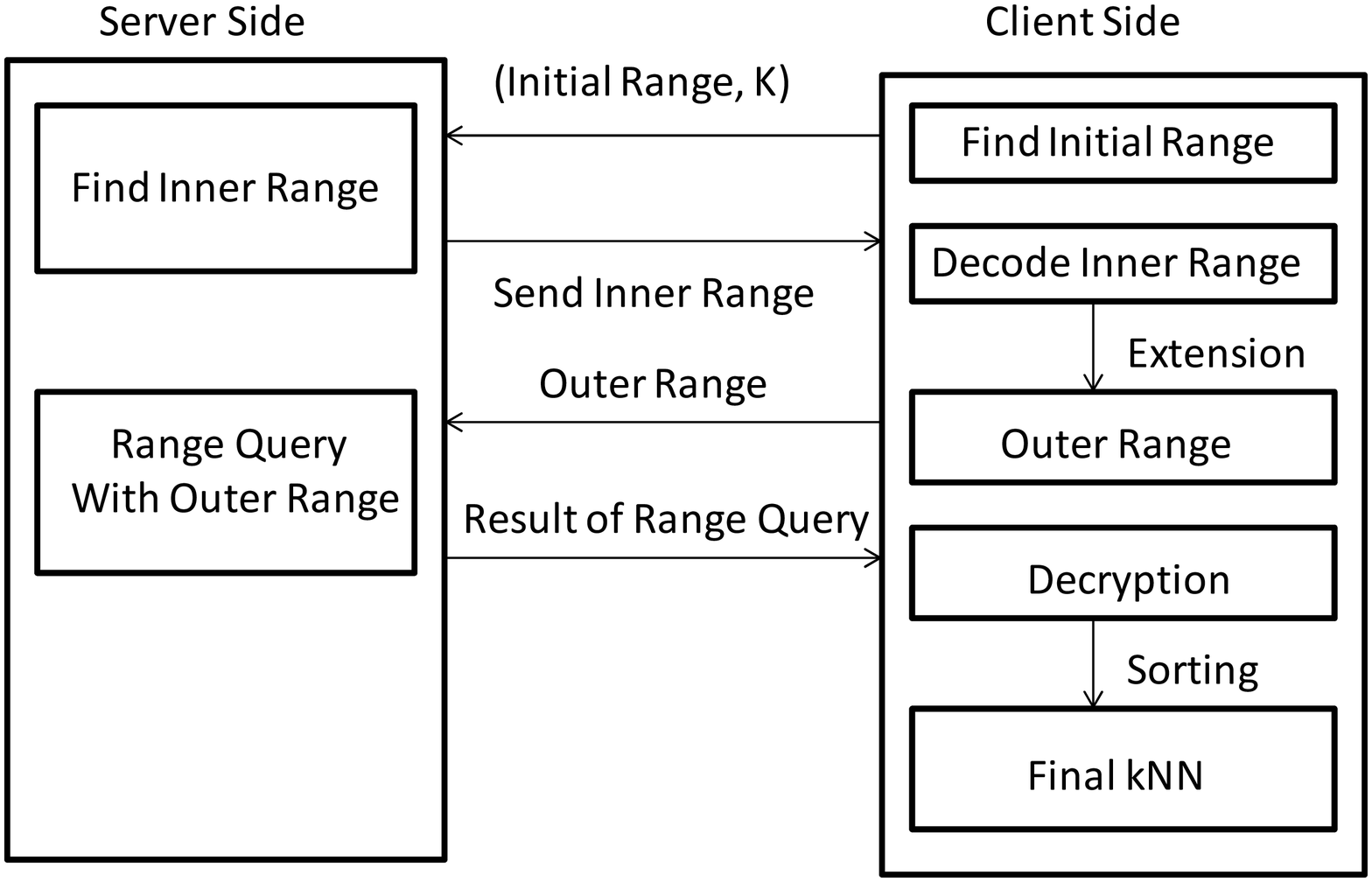}
\caption{Procedure of KNN-R algorithm}\label{fig:kNN-R-overview}
\end{minipage}
\hspace{0.1\linewidth}
\begin{minipage}{0.35\linewidth}
\centering
\includegraphics[width =\linewidth]{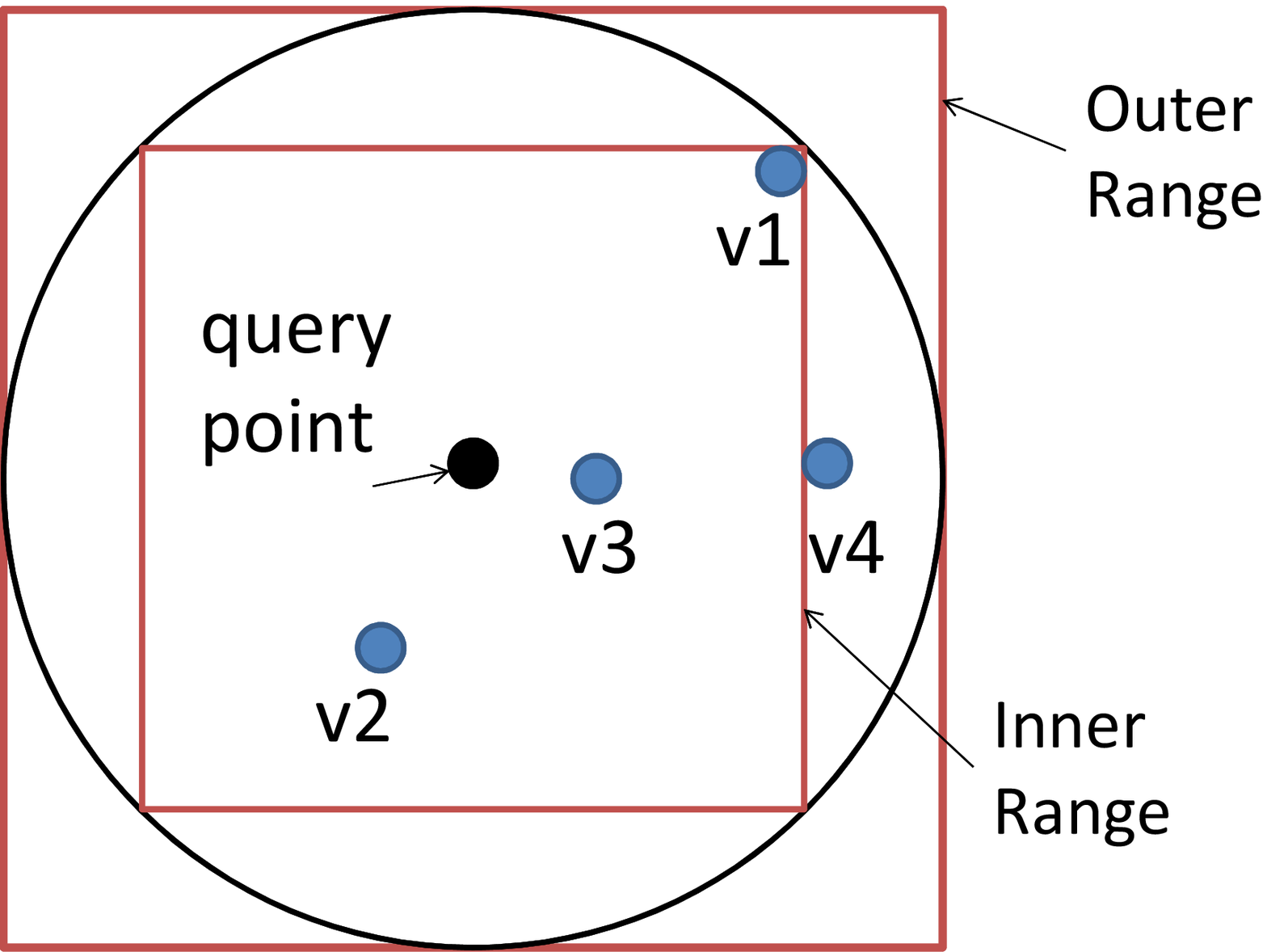}
\caption{Illustration for kNN-R Algorithm when k=3}\label{fig:kNN-R-K3}
\end{minipage}
\end{figure*}

The algorithm is based on \emph{square ranges} to approximately find the kNN candidates for a query point, which are defined as follows.
\begin{myDef}
A square range is a hyper-cube that is centered at the query point and with equal-length edges.
\end{myDef}
Figure \ref{fig:kNN-R-K3} illustrates the range-query-based kNN processing with two-dimensional data. The \emph{Inner Range} is the square range that contains at least $k$ points, and the
\emph{Outer Range} encloses the spherical range that encloses the inner range. The outer range surely
contains the kNN results (Proposition \ref{prop:kNN-R}) but it may also contain irrelevant points
that need to be filtered out.
\begin{myPro}\label{prop:kNN-R}
The kNN-R algorithm returns results with 100\% recall.
\end{myPro}
\begin{proof}
 The sphere in Figure \ref{fig:kNN-R-K3} between the outer range and the inner range covers all points with distances less than the radius $r$. Because the inner range contains at least $k$ points, there are at least $k$ nearest neighbors to the query points with distances less than the radius $r$. Therefore, the $k$ nearest neighbors must be in the outer range.
\end{proof}

The kNN-R algorithm consists of two rounds of interactions between the client and the server. Figure \ref{fig:kNN-R-overview} demonstrates the procedure. (1) The client will send the initial upper-bound range, which contains more than $k$ points, and the initial lower-bound range, which contains less than $k$ points, to the server. The server finds the inner range and returns to the client. (2) The client calculates the outer range based on the inner range and sends it back to the server. The server finds the records in the outer range and sends them to the client. (3) The client decrypts the records and find the top $k$ candidates as the final result.


If the points are approximately uniformly distributed, we can estimate the precision of the returned result. With the uniform assumption, the number of points in an area is proportional to the size of the area. If the inner range contains $m$ points, $m>=k$, the outer range contains $q$ points, and the dimensionality is $d$, we can derive
$q = 2^{d/2}m$.
Thus, the precision is $k/q = k/(2^{d/2}m)$. If $m\approx k$ and $d=2$, the precision is around 0.5. When $d$ increases, the precision decreases exponentially due to the curse of dimensionality \cite{marimont79}, which suggests kNN-R should not work effectively on high-dimensional data. We will show this weakness in experiments.


\subsection{Finding Compact Inner Square Range}
An important step in the kNN-R algorithm is to find the compact inner square range to achieve high precision. In the following, we give the $(k, \delta)$-range for efficiently finding the compact inner range.
\begin{myDef}
A \emph{$(k, \delta)$-range} is any square range centered at the query point, the number of points in which is in the range $[k, k+\delta]$, $\delta$ is a nonnegative integer.
\end{myDef}

We design an algorithm similar to binary search to efficiently find the $(k, \delta)$-range. Suppose a square range centered at the query point with length of $L$ in each dimension is represented as $S^{(L)}$. Let the number of points included by this range is $N^{(L)}$. If a square range $S^{(in)}$ is enclosed by another square range $S^{(out)}$, we say $S^{(in)} \subset S^{(out)}$. It directly follows that $N^{(in)}\leq N^{(out)}$, and also
\begin{corollary} \label{cor:enclose}
If $N^{(1)}< N^{(2)}$,  $S^{(1)} \subset S^{(2)}$.
\end{corollary}
Using this definition and notation, we can always construct a series of enclosed square ranges centered on the query point: $S^{(L_{1})}\subset S^{(L_{2})}\subset \ldots, \subset S^{(L_{m})}$. Correspondingly, the numbers of points enclosed by $\{S^{(L_i)}\}$ have the ordering $N^{(L_1)}\leq N^{(L_2)}\leq \ldots N^{(L_m)}$.
Assume that $S^{(L_{1})}$ is the initial range containing less than $k$ points and $S^{(L_{m})}$ is the initial upper bound range; both are sent by the client. The problem of finding the compact inner range $S$ can be mapped to a binary search over the sequence $\{S^{(L_i)}\}$.

In each step of the binary search, we start with a lower bound range, denoted as $S^{(low)}$ and a higher bound range, $S^{(high)}$. We want the corresponding numbers of enclosed points to satisfy
$N^{(low)} < k \leq N^{(high)}$ in each step, which is achieved with the following procedure.
First, we find the middle square range $S^{(mid)}$, where $mid = (low + high)/2$. If $S^{(mid)}$ covers no less than $k$ points, the higher bound: $S^{(high)}$ is updated to $S^{(mid)}$; otherwise, the lower bound: $S^{(low)}$ is updated to $S^{(mid)}$. At the beginning step $S^{(low)}$ is set to $S^{(L_{1})}$ and $S^{(high)}$ is $S^{(L_{m})}$. This process repeats until $N^{(mid)}<k+\delta$ or $high - low < \mathcal{E}$, where $\mathcal{E}$ is some small positive number. Algorithm \ref{alg:k-delta} in Appendix describes these steps.

\textbf{Selection of Initial Inner/Outer Bounds.}
The selection of initial inner bound can be the query point. If the query point is $q(q_1,\ldots, q_d)$, $S^{(L_{1})}$ is a hyper-cube defined by $\{ q_i \geq X_i\geq q_i, i=1\ldots d\}$. The naive selection of $S^{(L_m)}$ would be the whole domain. However, we can effectively reduce the range with a coarse density map organized in a tiny flat multidimensional tree, which can be included in the preprocessing step in the client side. The details will be ignored due to the space limitation.

\subsection{Finding Inner Range with RASP Perturbed Data}
Algorithm \ref{alg:k-delta} gives the basic ideas of finding the compact inner range in iterations. There are two critical operations in this algorithm: (1) finding the number of points in a square range and (2) updating the higher and lower bounds. Because range queries are secured in the RASP framework, the key is to update the bounds with the secured range queries, without the help of the client-side proxy server.

As discussed in the RASP query processing, a range query such as $S^{(L)}$ is encoded as the MBR$^{(L)}$ of its  polyhedron range in the perturbed space and the $2(d+2)$ dimensional conditions. $y^T\Theta_i^{(L)} y\leq 0$ determining the sides of the polyhedron, and each of the $d+2$ extended dimensions gets a pair of conditions for the upper and lower bounds, respectively.

The problem of binary range search is to use the higher bound range $S^{(high)}$ and the lower bound range $S^{(low)}$ to derive $S^{(mid)}$. When all of these ranges are secured, the problem is transformed to (1) deriving $\Theta_i^{(mid)}$ from $\Theta_i^{(high)}$ and $\Theta_i^{(low)}$; and (2) deriving MBR$^{(mid)}$ from MBR$^{(high)}$ and MBR$^{(low)}$.
The following discussion will be focused on the simplified RASP version without the OPE component, which will be extended with the OPE component.

We show that
\begin{myPro} \label{prop:theta}
\begin{equation*}
(\Theta_i^{(high)} + \Theta_{i}^{(low)})/2 = \Theta_{i}^{(mid)}.
\end{equation*}
\end{myPro}
\begin{proof}
Remember that $\Theta_{i}$ for $X_i<c_i$ can be represented as $(a_i-c_ia_{d+1})^T(v_0a_{d+1}-a_{d+2})$, where $a_i$ is the $i$-th row of the matrix $A$. Let the conditions be $X_i<h$, $X_i<l$, and $X_i<(h+l)/2$ for the high, low, and middle bounds, correspondingly.
Thus, $(\Theta_i^{(high)} + \Theta_{i}^{(low)})/2 = (a_i - ((h+l)/2)a_{d+1})^T(v_0a_{d+1}-a_{d+2})$, which is $\Theta_{i}^{(mid)}$.
\end{proof}

As we have mentioned, the MBR of an arbitrary polyhedron can be derived based on the vertices of the polyhedron. A polyhedron is mapped to another polyhedron after the RASP perturbation. Concretely, let a polyhedron $P$ has $m$ vertices $\{x_1,\ldots,x_m\}$, which are mapped to the vertices in the perturbed space: $\{y_1,\ldots,y_m\}$. Then, the upper bound and lower bound of dimension  $j$ of the $MBR$ of the polyhedron in the perturbed space are  determined by $\max\{ y_{ij}, i=1\ldots m\}$ and $\min\{y_{ij}, i=1\ldots m\}$, respectively.

Let the j-th dimension of $MBR^{(L)}$ represented as $[s_{j,min}^{(L)}, s_{j,max}^{(L)}]$, where $s_{j,min}^{(L)}$ $= \min\{ y_{ij}^{(L)}, i=1\ldots m\}$, and $s_{j,max}^{(L)} = \max\{ y_{ij}^{(high)}, i=1\ldots m\}$. Now we choose the $MBR^{(MID)}$ as follows: for j-th dimension we use $[(s_{j,min}^{(low)}+s_{j,min}^{(high)})/2, (s_{j,max}^{(low)}+s_{j,max}^{(high)})/2]$. We show that
\begin{myPro}
$MBR^{(MID)}$ encloses $MBR^{(mid)}$.
\end{myPro}
The details of proof can be found in Appendix. Because the MBR is only used for the first stage of range query processing, a slightly larger MBR still encloses the polyhedron, which guarantees the correctness of the two-stage range query processing.

\textbf{Including the OPE component.}
The results on $\Theta_i^{(mid)}$ and MBR$^{(MID)}$ can be extended to the RASP scheme with the OPE component. However, due to the introduction of the order preserving function $f_i()$, the middle point may not be strictly the middle point, but somewhere between the higher bound and lower bound. We use ``between''(btw) to denote it.

Specifically, if $X_i<h$ and $X_i<l$ are the corresponding conditions for the higher and lower bounds. Let the condition for the ``between'' bound be $X_i<b$ that satisfies $f_i(b) = (f_i(h)+f_i(l))/2$. According to the OPE property, we have $l<b<h$, i.e., the corresponding range is still between the lower range and higher range. Therefore, the same binary search algorithm can still be applied, according to Corollary \ref{cor:enclose}. The server can also derive $(\Theta_i^{(high)} + \Theta_{i}^{(low)})/2 = (a_i - ((f_i(h)+f_i(l))/2)a_{d+1})^T(v_0a_{d+1}-a_{d+2}) = \Theta_i^{btw}$, a result similar to Proposition \ref{prop:theta}.

Similarly, we define MBR$^{(BTW)}$ with $f_i(s_{i,max}^{(BTW)}) = (f_i(s_{i,max}^{(low)})+f_i(s_{i,max}^{(high)}))/2$ and $f_i(s_{i, min}^{(BTW)}) = (f_i(s_{i,min}^{(low)})+f_i(s_{i,min}^{(high)}))/2$, while MBR$^{(btw)}$ is defined based on the vertices to be consistent with $\Theta_i^{(btw)}$. Because the relationships Eq. \ref{eq:max} and \ref{eq:min} in Appendix are still true with the OPE transformation $f_i()$, we can prove that MBR$^{(BTW)}$ also encloses MBR$^{(btw)}$. Due to the space limitation, we skip the details.

\subsection{Defining Initial Bounds} \label{sec:prep}
The complexity of the $(k,\delta)$-range algorithm is determined by the initial bounds provided by the client. Thus, it is important to
provide compact ones to help the server process queries more efficiently. The initial lower bound is defined as the query point. For $q(q_1,\ldots,q_d)$, the dimensional bounds are simply $q_j\leq X_j\leq q_j$.

The higher bounds can be defined in multiple ways. (1) Applications often have a user-specified interest bound, for example, returning the nearest gas station in 5 miles, which can be used to define the higher bound. (2) We can also use center-distance based bound setting. Let the query point has a distance $\gamma$ to the distribution center - as we always work on normalized distributions, the center is $(0,\ldots, 0)$. The upper bound is defined as $q_j-\epsilon\gamma\leq X_j\leq q_j+\epsilon\gamma$, where $epsilon \in (0, 1]$ defines the level of conservativity. (3) If it is really expected to include all candidate kNN regardless how distant they are, we can include a rough density-map (a multidimensional histgram) for quickly identifying the appropriate higher bound. However, this method works best for low dimensional data as the number of bins exponentially increases with the number of dimensions. In experiments, we simply use the method (1) and 5\% of the domain length for the extension.

\subsection{Security of kNN Queries}
As all kNN queries are completely transformed to range queries, the security of kNN queries are equivalent to the security of range queries. According to the previous discussion in Section \ref{sec:QuerySecurity}, the transformed range queries are secure under the assumptions. Therefore, the kNN queries are also secure. Detailed proofs have to be skipped for space limitation.

\section{Experiments}\label{sec:exp}
In this section, we present four sets of experimental results to investigate the following questions, correspondingly. (1) How expensive is the RASP perturbation? (2) How resilient the OPE enhanced RASP is to the ICA-based attack? (3) How efficient is the two-stage range query processing? (4) How efficient is the kNN-R query processing and what are the advantages?

\subsection{Datasets}
Three datasets are used in experiments. (1) A synthetic dataset that draws samples from uniform distribution in the range [0, 1]. (2) The Adult dataset from UCI machine learning database\footnote{http://archive.ics.uci.edu/ml/}. We assign numeric values to the categorical values using a simple one-to-one mapping scheme, as described in Section \ref{sec:RASP}. (3) The 2-dimensional NorthEast location data from rtreeportal.org.

\begin{figure*}[tbh]
\centering
\begin{minipage}{0.45\linewidth}
\centering
\includegraphics[width=\linewidth]{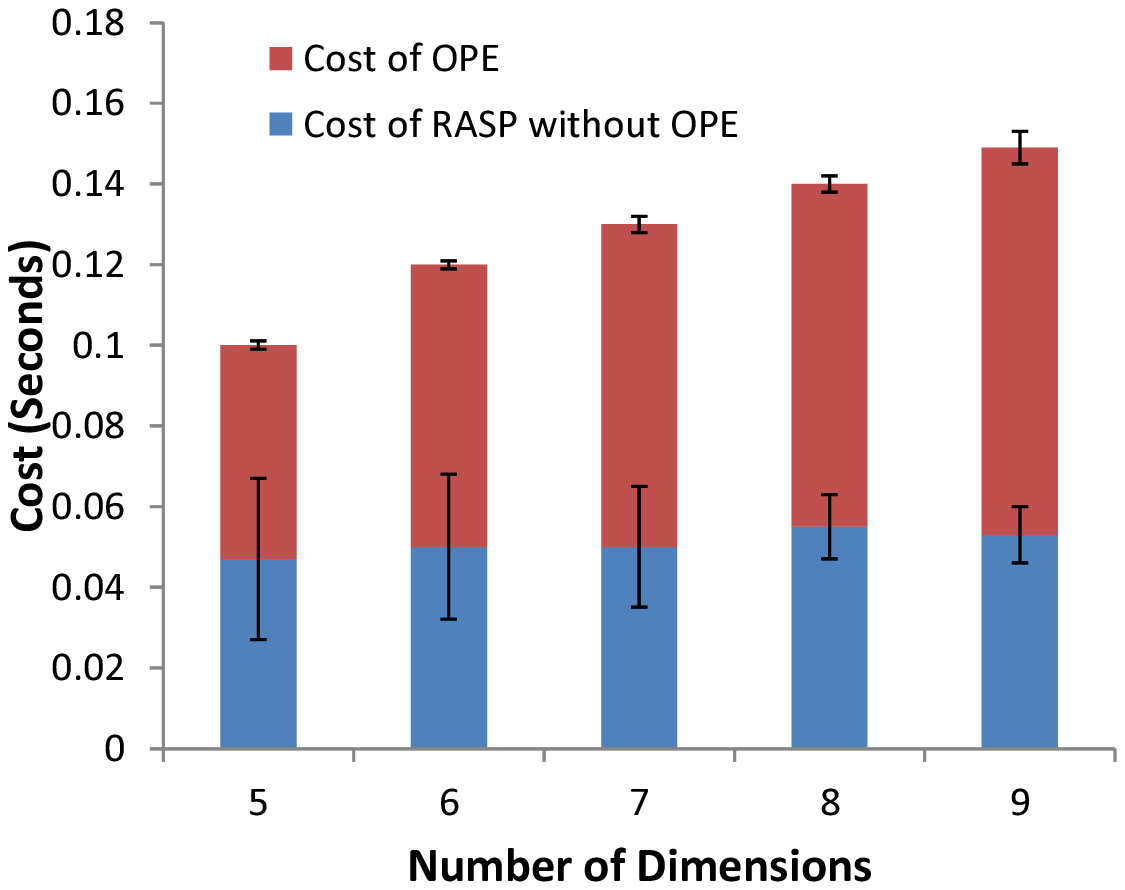}
\caption{The cost distribution of the full RASP scheme. Data: Adult (20K records,5-9 dimensions)}\label{fig:enc-cost}
\end{minipage}
\hspace{0.08\linewidth}
\begin{minipage}{0.45\linewidth}
\centering
\includegraphics[width=\linewidth]{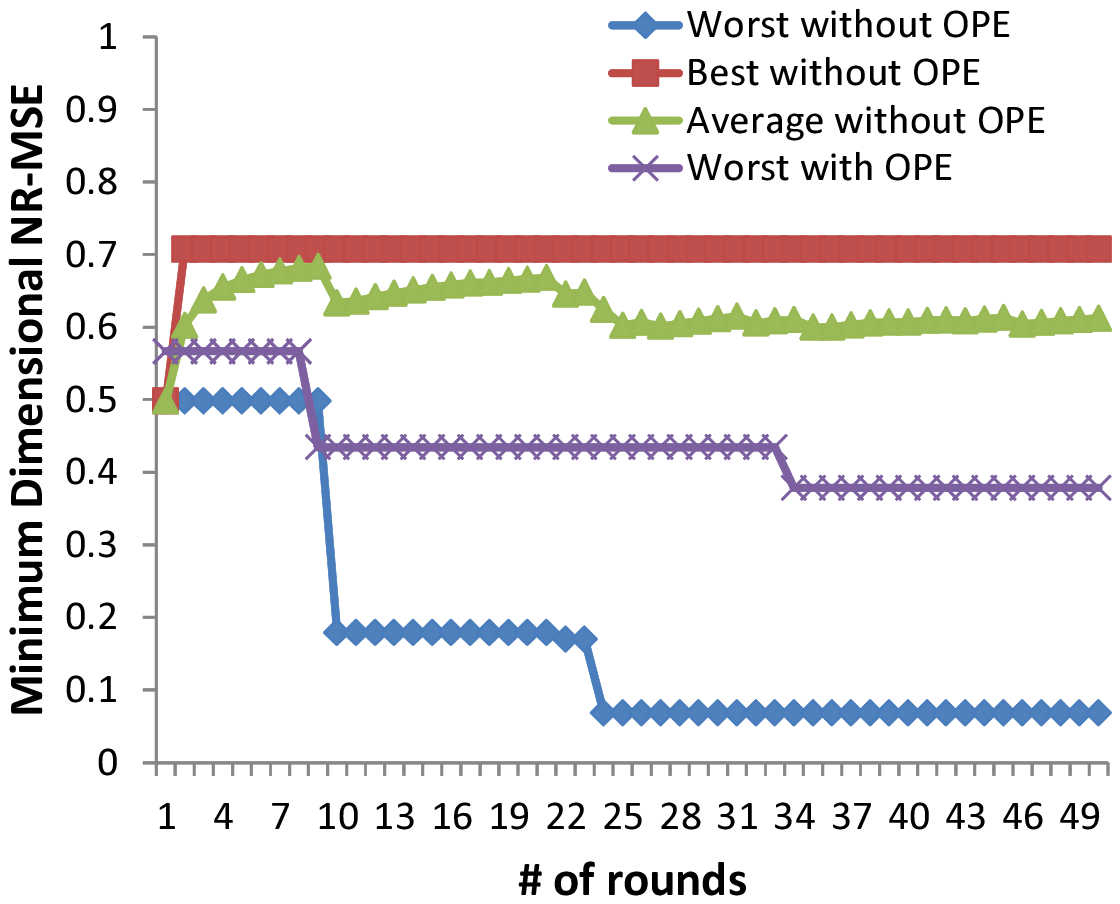}
\caption{Randomly generated matrix $A$ and the progressive resilience to ICA attack. Data: Adult (10 dimensions, 10K records)}\label{fig:adult-ica}
\end{minipage}
\end{figure*}

\subsection{Cost of RASP Perturbation}
In this experiment, we study the costs of the components in the RASP perturbation. The major costs can be divided into two parts: the OPE and the rest part of RASP. We implement a simple OPE scheme \cite{rakesh04} by mapping original column distributions to normal distributions. The OPE algorithm partitions the target distribution into buckets. Then, the sorted original values are proportionally partitioned according to the target bucket distribution to create the buckets for the original distribution. With the aligned original and target buckets, an original value can be mapped to the target bucket and appropriately scaled. Therefore, the encryption cost mainly comes from the bucket search procedure (proportional to $\log D$, where $D$ is the number of buckets).  Figure \ref{fig:enc-cost} shows the cost distributions for 20K records at different number of dimensions. The dimensionality has slight effects on the cost of RASP perturbation. Overall, the cost of processing 20K records is only around 0.1 second.

\subsection{Resilience to ICA Attack}
We have discussed the methods for countering the ICA distributional attack on the perturbed data. In this set of experiments, we evaluate how resilient the RASP perturbation is to the distributional attack.

\textbf{Results.} We simulate the ICA attack for randomly chosen matrices $A$. The data used in the experiment is the 10-dimensional Adult data with 10K records. Figure \ref{fig:adult-ica} shows the progressive results in a number of randomly chosen matrices $A$. The x-axis represents the total number of rounds for randomly choosing the matrix $A$; the y-axis represents the minimum dimensional NR\_MSE among all dimension.  Without OPE, the label ``Best-without-OPE'' represents the most resilient $A$ at the round $i$, ``Worst-without-OPE'' represents the $A$ of the weakest resilience, and ``Average-without-OPE'' is the average quality of the generated $A$ matrices for $i$ rounds. We see that the best case is already close to the upper bound 0.7 (Section \ref{sec:DataAttack}). With the OPE component, the worst case can also be significantly improved.

\subsection{Performance of Two-stage Range Query Processing}
In this set of experiments, we study the performance aspects of polyhedron-based range query processing. We use the two-stage processing strategy described in Section \ref{sec:RANGE}, and explore the additional cost incurred by this processing strategy. We implement the two-stage query processing based on an R*tree implementation provided by Dr. Hadjieleftheriou at AT\&T Lab\footnote{http://www2.research.att.com/~marioh/spatialindex/}. The block size is 4KB and we allow each block to contain only 20 entries to mimic a large database with many disk blocks. Samples from the original databases in different size (10,000 $-$ 50,000 records, i.e., 500-2500 data blocks) are perturbed and indexed for query processing. Another set of indices is also built on the original data for the performance comparison with non-perturbed query processing. We will use the number of disk block accesses, including index blocks and data blocks, to assess the performance to avoid the possible variation caused by other parts of the computer system. In addition, we will also show the wall-clock time for some results.

Recall the two-stage processing strategy: using the MBR to search the indexing tree, and filtering the returned result with the secured query in quadratic form. We will study the performance of the first stage by comparing it to two additional methods: (1) the original queries with the index built on the original data, which is used to identify how much additional cost is paid for querying the MBR of the transformed query; (2) the linear scan approach, which is the worst case cost. Range queries are generated randomly within the domain of the datasets, and then transformed with the method described in the Section \ref{sec:RANGE}. We also control the range of the queries to be [10\%,20\%,30\%,40\%,50\%] of the total range of the domain, to observe the effect of the scale of the range to the performance of query processing.

\begin{figure*}[tbh]
\begin{minipage}{\linewidth}
\centering
\subfigure{\includegraphics[width = 0.3\linewidth]{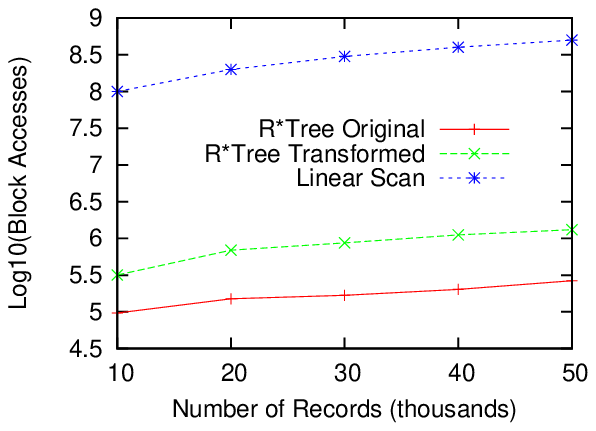}}
\subfigure{\includegraphics[width = 0.3\linewidth]{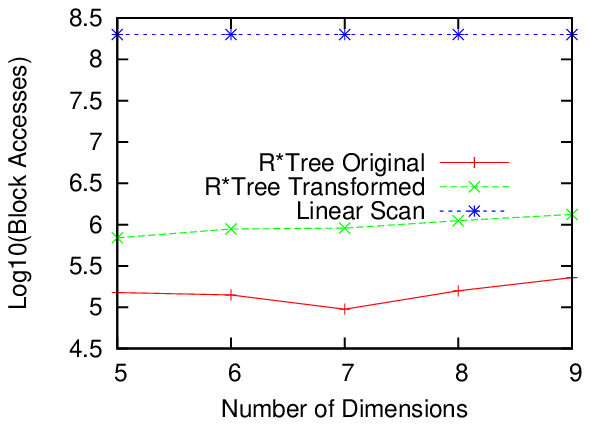}}
\subfigure{\includegraphics[width = 0.3\linewidth]{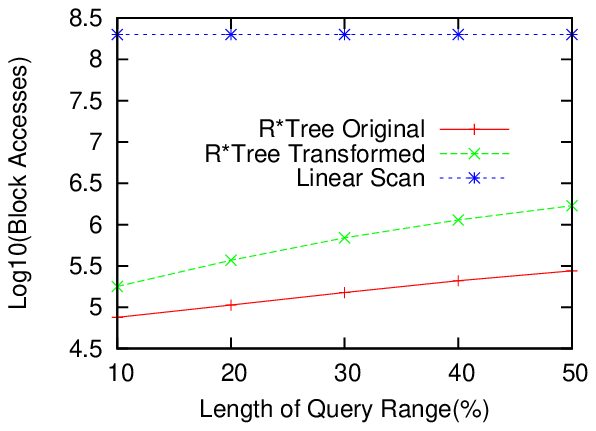}}
\caption{Performance comparison on Uniform data. Left: data size vs. cost of query; Middle: data dimensionality vs. cost of query; Right: query range (percentage of the domain) vs. cost of query}\label{fig:uniform-perf}
\end{minipage}
\end{figure*}

\begin{figure*}[tbh]
\begin{minipage}{\linewidth}
\centering
\subfigure{\includegraphics[width = 0.3\linewidth]{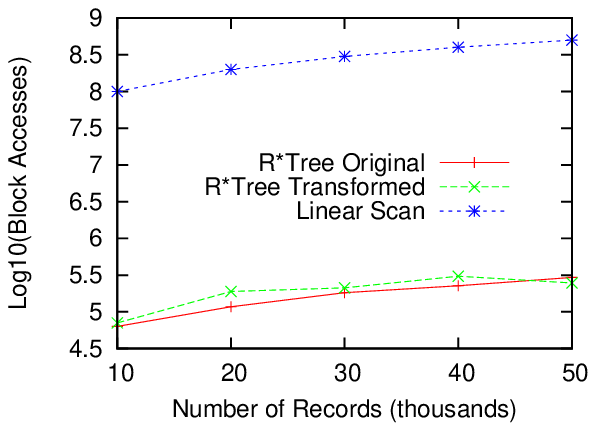}}
\subfigure{\includegraphics[width = 0.3\linewidth]{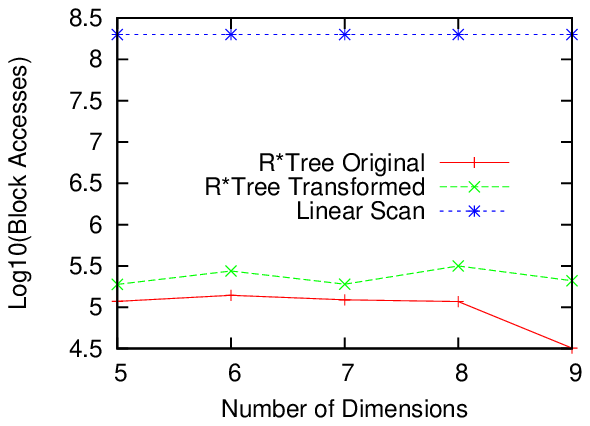}}
\subfigure{\includegraphics[width = 0.3\linewidth]{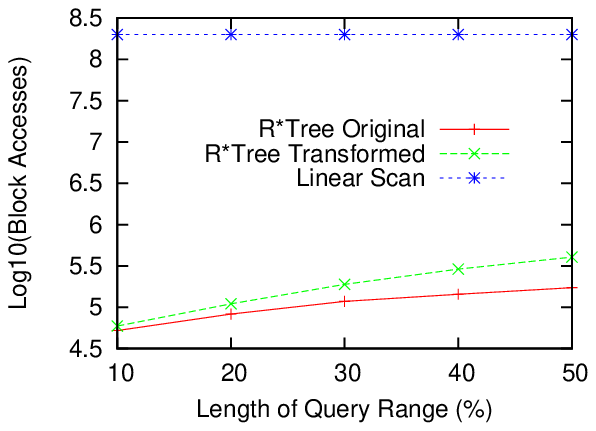}}
\caption{Performance comparison on Adult data. Left: data size vs. cost of query; Middle: data dimensionality vs. cost of query; Right: query range (percentage of the domain) vs. cost of query}\label{fig:adult-perf}
\end{minipage}
\end{figure*}

\noindent\textbf{Results. }The first pair of figures (the left subfigures of Figure \ref{fig:uniform-perf} and \ref{fig:adult-perf}) shows the number of block accesses for 10,000 queries on different sizes of data with different query processing methods. For clear presentation, we use $\log_{10}$(\# of block accesses) as the y-axis. The cost of linear scan is simply the number of blocks for storing the whole dataset. The data dimensionality is fixed to 5 and the query range is set to 30\%  of the whole domain. Obviously, the first stage with MBR for polyhedron has a cost much cheaper than the linear scan method and only moderately higher than R*tree processing on the original data. Interestingly, different distributions of data result in slightly different patterns. The costs of R*tree on transformed queries are very close to those of original queries for Adult data, while the gap is larger on uniform data. The costs over different dimensions and different query ranges show similar patterns.

\begin{table}[tbh]
\centering
\scriptsize
\begin{tabular}{|c|c|c|c|c|c|c|c|}
\hline
  &Linear Scan&R*Tree-Orig&PrepQ&Stage-1& Stage-2& rpq & purity \\
\hline
 Uniform5D & 21.12 & 0.27 &0.007 &4.19 & 0.01 &51.92 &7.76\%\\
Adult5D & 16.28 & 0.39 & 0.007& 1.9 &0.01& 5.12 &1.17\%\\
\hline
\end{tabular}
\caption{Wall clock cost distribution (milliseconds) and comparison.} \label{tab:perf}
\normalsize
\end{table}

We also studied the cost of the second stage.  We use ``PrepQ'' to represent the client-side cost of transforming queries, ``purity'' to represent the rate (final result count)/(1st stage result count), and records per query (``RPQ'') to represent the average number of records per query for the first stage results. The quadratic filtering conditions are used in experiments. Table \ref{tab:perf} compares the average wall-clock time (milliseconds) per query for the two stages, the RPQ values for stage 1, and the purity of the stage-1 result. The tests are run with the setting of 10K  queries, 20K records, 30\% dimensional query range and 5 dimensions. Since the 2nd stage is done in memory, its cost is much lower than the 1st-stage cost. Overall, the two stage processing is much faster than linear scan and comparable to the original R*Tree processing.

\subsection{Performance of kNN-R Query Processing}
In this set of experiments, we investigate several aspects of kNN query processing. (1) We will study the cost of (k, $\delta$)-Range algorithm, which mainly contributes to the server-side cost. (2) We will show the overall cost distribution over the cloud side and the proxy server. (3) We will show the advantages of kNN-R over another popular approach: the Casper approach \cite{mokbel06} for privacy-preserving kNN search.

\textbf{(k, $\delta$)-Range Algorithms}
In this set of experiments, we want to understand how the setting of the $\delta$ parameter affects the performance and the result precision. Figure \ref{fig:delta} shows the effect of $\delta$ setting to the $(k,\delta)$-range algorithm. Both datasets are two-dimensional data. As $\delta$ becomes larger, both the precision and the number of rounds needs to reach the $\delta$ condition decreases. Note that each round corresponds to one server-side range query. The choice of $\delta$ represents a tradeoff between the precision and the performance.

\begin{figure}[tbh]
\begin{minipage}{.65\linewidth}
\centering
\includegraphics[width =\linewidth]{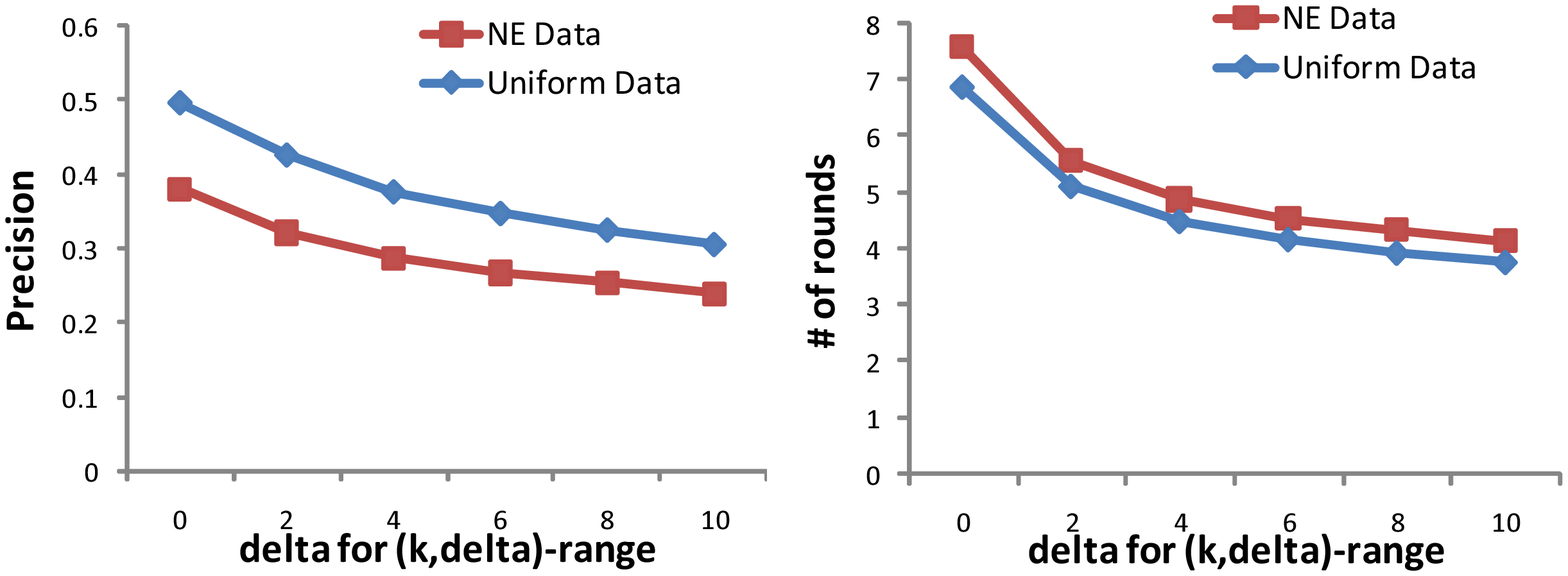}
\caption{Performance and result precision for different $\delta$ setting of the $(k,\delta)$-range algorithm for 2-dimensional data.}\label{fig:delta}
\end{minipage}
\hspace{0.04\linewidth}
\begin{minipage}{.3\linewidth}
\centering
\includegraphics[width =\linewidth]{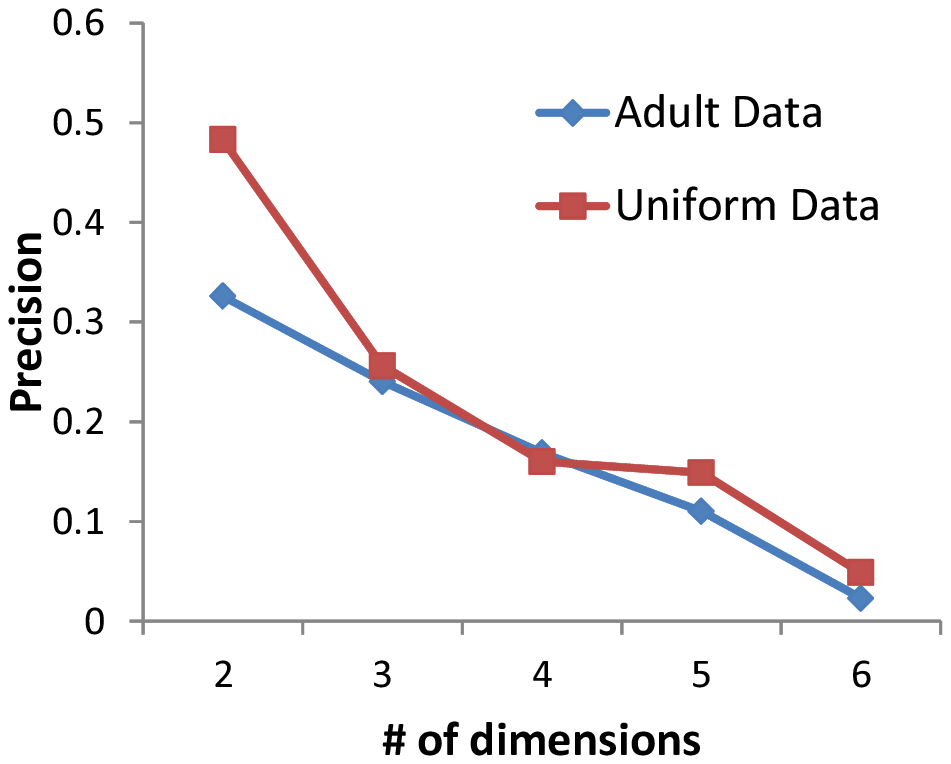}
\caption{Precision reduction with more dimension.}\label{fig:dim-precision}
\end{minipage}
\end{figure}


As we have discussed, the major weakness with the kNN-R algorithm is the precision reduction with increased dimensionality. When the dimensionality increases, the precision can significantly drop, which will increase the cost of post-processing in the client side. Figure \ref{fig:dim-precision} shows this phenomenon with the real Adult data and the simulated uniform data. However, compared to the overall cost, the client-side cost increase is still acceptable. We will show the comparison next.

\textbf{Overall Costs.}
Many secure approaches cannot use indices for query processing, which results in poor performance. For example, the secure dot-product approach \cite{wong09} encodes the points with random projections and recovers dot-products in query processing for distance comparison. The way of encoding data disallows the index-based query processing. Without the aid of indices, processing a kNN query will have to scan the entire database, leaving many optimization impossible to implement. 

One concern with the kNN-R approach is the workload on the proxy server. Different from range query, the proxy server will need to filter out the points returned by the server to find the final kNN. A reduced precision due to the increased dimensionality will imply an increased burden for the proxy server. We need to show how significant this proxy cost is.

We use the database of 100 thousands of data points and 1000 randomly selected queries for the 1NN experiment.  The wall clock time (milliseconds) is used to show the average cost per query in Table \ref{tab:knn}. We also list the cost of the secure dot-product method \cite{wong09} for comparison. Table \ref{tab:knn} shows that the proxy server takes a negligible pre-processing cost and a very small post-processing cost, even for reduced precision in the 5D datasets. We use 5\% domain length to extend the query point to form the initial higher bound. Compared to the dot-product method, the user-specified higher bound setting can cut off uninteresting regions, giving significant performance gain for sparse or skewed datasets, such as Adult5D. This cut-off effect cannot be implemented with the dot-product method. Furthermore, even for dense cases like the 2D datasets, the overall cost is only about half of the dot-product method.

\begin{table}[tbh]
\centering
\scriptsize
\begin{tabular}{|c|c|c|c|c|}
\hline
 Data\& setting & Liner Scan &Pre-processing & Server Cost & Post-processing\\
\hline
Uniform2D/kNN-R & 27.37 &0.01 &13.54&0.04 \\
Adult2D/kNN-R  & 26.09 &0.01 &14.48&0.06\\
\hline
Uniform5D/kNN-R & 33.03 &0.01 &13.79&0.34\\
Adult5D/kNN-R & 31.96 &0.01& 2.56 &0.05\\
\hline
\end{tabular}
\caption{Per-query performance comparison (milliseconds) between linear scan on the original non-perturbed data and index-aided kNN-R processing on perturbed data. 
} \label{tab:knn}
\normalsize
\end{table}

\textbf{Comparing kNN-R with the Casper Approach.}
In this set of experiments, we compare our approach and the Casper approach with a focus on the tradeoff between the data confidentiality and the query result precision (which indicates the workload of the in-house proxy). Based on the description in the paper \cite{mokbel06}, we implement the 1NN query processing algorithm for the experiment.

The Casper approach uses cloaking boxes to hide both the original data points in the database and the query points. It can also use the index to process kNN queries. The confidentiality of data in Casper is solely defined by the size of cloaking box. Roughly speaking, the actual point has the same probability to be anywhere in the cloaking box. However, the size of cloaking box also directly affects the precision of query results. Thus, the decision on the box size represents a tradeoff between the precision of query results and the data confidentiality.

For clear presentation, we assume each dimension has the same length of domain, $h$ and each cloaking box is square with an edge-length $e$. Assume the whole domain also has a uniform distribution. According to the variance of uniform distribution, the NR\_MSE measure is $\sqrt{6}e/(3h)$.  To achieve the protection of 10\% domain length, we have $e\approx 0.12 h$.

\begin{figure}
\centering
\includegraphics[width=.4\linewidth]{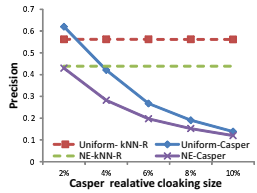}
\caption{The impact of cloaking-box size on precision for Casper for the NE data.}\label{fig:casper}
\end{figure}

In Figure \ref{fig:casper}, the x-axis represents NR\_MSE, i.e., the Casper's relative cloaking-edge length. It shows that when the edge length is increased from 2\% to 10\%, the precision dramatically drops from 62\% to 13\% for the 2D uniform data and 43\% to 10\% for the 2D NE data, which shows the severe conflict between precision and confidentiality. The kNN-R's results are also shown for comparison.

\section{Related work}
\label{sec:related-work}

\subsection{Protecting Outsourced Data}
\textbf{Order Preserving Encryption. }
Order preserving encryption (OPE) \cite{rakesh04} preserves the dimensional value order after encryption. It can be described as a function $y=F(x), \forall x_i, x_j, x_i<(>,=)x_j \Leftrightarrow y_i<(>,=)y_j$.  A well-known attack is based on attacker's prior knowledge on the original distributions of the attributes. If the attacker knows the original distributions and manages to identify the mapping between the original attribute and its encrypted counterpart, a bucket-based distribution alignment can be performed to break the encryption for the attribute \cite{chen11rasp}. 
There are some applications of OPE in outsourced data processing. For example, Yiu et al. \cite{yiu10} uses a hierarchical space division method to encode spatial data points, which preserves the order of dimensional values and thus is one kind of OPE. 

\textbf{Crypto-Index. }
Crypto-Index is also based on column-wise bucketization. It assigns a random ID to each bucket; the values in the bucket are replaced with the bucket ID to generate the auxiliary data for indexing. To utilize the index for query processing, a normal range query condition has to be transformed to a set-based query on the bucket IDs. For example, $X_i<a_i$ might be replaced with $X_i' \in [ID_1,ID_2,ID3]$.
A bucket-diffusion scheme \cite{hore04} was proposed to protect the access pattern, which, however, has to sacrifice the precision of query results, and thus increase the client's cost of filtering the query result.

\textbf{Distance-Recoverable Encryption.} DRE is the most intuitive method for preserving the nearest neighbor relationship. Because of the exactly preserved distances, many attacks can be applied \cite{wong09,liu06pkdd,keke07sdm}. Wong et al. \cite{wong09} suggest preserving dot products instead of distances to find kNN, which is more resilient to distance-targeted attacks. One drawback is the search algorithm is limited to linear scan and no indexing method can be applied. 

\subsection{Preserving Query Privacy}
Private information retrieval (PIR) \cite{chor98} tries to fully preserve the privacy of access pattern, while the data may not be encrypted. PIR schemes are normally very costly. Focusing on the efficiency side of PIR, Williams et al. \cite{williams08} use a pyramid hash index to implement efficient privacy preserving data-block operations based on the idea of Oblivious RAM. It is different from our setting of high throughput range query processing.

Hu et al. \cite{hu11} addresses the query privacy problem and requires the authorized query users, the data owner, and the cloud to collaboratively process kNN queries. However, most computing tasks are done in the user's local system with heavy interactions with the cloud server. The cloud server only aids query processing, which does not meet the principle of moving computing to the cloud.

Papadopoulos et al. \cite{papadopoulos10} uses private information retrieval methods \cite{chor98} to enhance location privacy. However, their approach does not consider protecting the confidentiality of data.
SpaceTwist \cite{yiu08} proposes a method to query kNN by providing a fake user's location for preserving location privacy. But the method does not consider data confidentiality, as well. The Casper approach \cite{mokbel06} considers both data confidentiality and query privacy, the detail of which has been discussed in our experiments.

\subsection{Other Related Work}
Another line of research \cite{shi07} facilitates authorized users to access only the authorized portion of data, e.g., a certain range, with a public key scheme. However, the underlying encryption schemes do not produce indexable encrypted data. The setting of multidimensional range query in \cite{shi07} is different from ours. Their approach requires that the data owner provides the indices and keys for the server, and authorized users use the data in the server.  While in the cloud database scenario, the cloud server takes more responsibilities of indexing and query processing. Secure keyword search on encrypted documents \cite{reza06,wang10icdcs,cao11} scans each encrypted document in the database and finds the documents containing the keyword, which is more like point search in database.
The research on privacy preserving data mining has discussed multiplicative perturbation methods \cite{keke11kais}, which are similar to the RASP encryption, but with more emphasis on preserving the utility for data mining. 

\section{Conclusion}
We propose the RASP perturbation approach to hosting query services in the cloud, which satisfies the CPEL criteria: data Confidentiality, query Privacy, Efficient query processing, and Low in-house workload. The requirement on low in-house workload is a critical feature to fully realize the benefits of cloud computing, and efficient query processing is a key measure of the quality of query services.

RASP perturbation is a unique composition of OPE, dimensionality expansion, random noise injection, and random projection, which provides unique security features. It aims to preserve the topology of the queried range in the perturbed space, and allows to use indices for efficient range query processing.
With the topology-preserving features, we are able to develop efficient range query services to achieve sub-linear time complexity of processing queries. We then develop the kNN query service based on the range query service. The security of both the perturbed data and the protected queries is carefully analyzed under a precisely defined threat model. We also conduct several sets of experiments to show the efficiency of query processing and the low cost of in-house processing.

We will continue our studies on two aspects: (1) further improve the performance of query processing for both range queries and kNN queries; (2) formally analyze the leaked query and access patterns and the possible effect on both data and query confidentiality.

\bibliographystyle{IEEETrans}
\bibliography{paper}

\begin{thebibliography}{10}
\providecommand{\url}[1]{#1}
\csname url@samestyle\endcsname
\providecommand{\newblock}{\relax}
\providecommand{\bibinfo}[2]{#2}
\providecommand{\BIBentrySTDinterwordspacing}{\spaceskip=0pt\relax}
\providecommand{\BIBentryALTinterwordstretchfactor}{4}
\providecommand{\BIBentryALTinterwordspacing}{\spaceskip=\fontdimen2\font plus
\BIBentryALTinterwordstretchfactor\fontdimen3\font minus
  \fontdimen4\font\relax}
\providecommand{\BIBforeignlanguage}[2]{{%
\expandafter\ifx\csname l@#1\endcsname\relax
\typeout{** WARNING: IEEEtranS.bst: No hyphenation pattern has been}%
\typeout{** loaded for the language `#1'. Using the pattern for}%
\typeout{** the default language instead.}%
\else
\language=\csname l@#1\endcsname
\fi
#2}}
\providecommand{\BIBdecl}{\relax}
\BIBdecl

\bibitem{rakesh04}
R.~Agrawal, J.~Kiernan, R.~Srikant, and Y.~Xu, ``Order preserving encryption
  for numeric data,'' in \emph{{Proceedings of ACM SIGMOD Conference}}, 2004.

\bibitem{cloud09}
M.~Armbrust, A.~Fox, R.~Griffith, A.~D. Joseph, R.~K. andAndy Konwinski,
  G.~Lee, D.~Patterson, A.~Rabkin, I.~Stoica, and M.~Zaharia, ``Above the
  clouds: A berkeley view of cloud computing,'' \emph{Technical Report,
  University of Berkerley}, 2009.

\bibitem{bau11}
J.~Bau and J.~C. Mitchell, ``Security modeling and analysis,'' \emph{IEEE
  Security and Privacy}, vol.~9, no.~3, pp. 18--25, 2011.

\bibitem{boyd04}
S.~Boyd and L.~Vandenberghe, \emph{Convex Optimization}.\hskip 1em plus 0.5em
  minus 0.4em\relax Cambridge University Press, 2004.

\bibitem{cao11}
N.~Cao, C.~Wang, M.~Li, K.~Ren, and W.~Lou, ``Privacy-preserving multi-keyword
  ranked search over encrypted cloud data,'' in \emph{INFOCOMM}, 2011.

\bibitem{chen11rasp}
K.~Chen, R.~Kavuluru, and S.~Guo, ``Rasp: Efficient multidimensional range
  query on attack-resilient encrypted databases,'' in \emph{ACM Conference on
  Data and Application Security and Privacy}, 2011, pp. 249--260.

\bibitem{keke11kais}
K.~Chen and L.~Liu, ``Geometric data perturbation for outsourced data mining,''
  \emph{{Knowledge and Information Systems}}, 2011.

\bibitem{keke07sdm}
K.~Chen, L.~Liu, and G.~Sun, ``Towards attack-resilient geometric data
  perturbation,'' in \emph{SIAM Data Mining Conference}, 2007.

\bibitem{chor98}
B.~Chor, E.~Kushilevitz, O.~Goldreich, and M.~Sudan, ``Private information
  retrieval,'' \emph{ACM Computer Survey}, vol.~45, no.~6, pp. 965--981, 1998.

\bibitem{reza06}
R.~Curtmola, J.~Garay, S.~Kamara, and R.~Ostrovsky, ``Searchable symmetric
  encryption: improved definitions and efficient constructions,'' in
  \emph{Proceedings of the 13th ACM conference on Computer and communications
  security}.\hskip 1em plus 0.5em minus 0.4em\relax New York, NY, USA: ACM,
  2006, pp. 79--88.

\bibitem{draper98}
N.~R. Draper and H.~Smith, \emph{Applied Regression Analysis}.\hskip 1em plus
  0.5em minus 0.4em\relax Wiley, 1998.

\bibitem{hakan02sigmod}
H.~Hacigumus, B.~Iyer, C.~Li, and S.~Mehrotra, ``Executing sql over encrypted
  data in the database-service-provider model,'' in \emph{{Proceedings of ACM
  SIGMOD Conference}}, 2002.

\bibitem{hastie01}
T.~Hastie, R.~Tibshirani, and J.~Friedman, \emph{The Elements of Statistical
  Learning}.\hskip 1em plus 0.5em minus 0.4em\relax Springer-Verlag, 2001.

\bibitem{hore04}
B.~Hore, S.~Mehrotra, and G.~Tsudik, ``A privacy-preserving index for range
  queries,'' in \emph{{Proceedings of Very Large Databases Conference (VLDB)}},
  2004.

\bibitem{hu11}
H.~Hu, J.~Xu, C.~Ren, and B.~Choi, ``Processing private queries over untrusted
  data cloud through privacy homomorphism,'' \emph{{Proceedings of IEEE
  International Conference on Data Engineering (ICDE)}}, pp. 601--612, 2011.

\bibitem{huang05}
Z.~Huang, W.~Du, and B.~Chen, ``Deriving private information from randomized
  data,'' in \emph{{Proceedings of ACM SIGMOD Conference}}, 2005.

\bibitem{hyvarinen01}
A.~Hyvarinen, J.~Karhunen, and E.~Oja, \emph{Independent Component
  Analysis}.\hskip 1em plus 0.5em minus 0.4em\relax Wiley, 2001.

\bibitem{jolliffe86}
I.~T. Jolliffe, \emph{Principal Component Analysis}.\hskip 1em plus 0.5em minus
  0.4em\relax Springer, 1986.

\bibitem{liff06}
F.~Li, M.~Hadjieleftheriou, G.~Kollios, and L.~Reyzin, ``Dynamic authenticated
  index structures for outsourced databases,'' in \emph{{Proceedings of ACM
  SIGMOD Conference}}, 2006.

\bibitem{liu06pkdd}
K.~Liu, C.~Giannella, and H.~Kargupta, ``An attacker's view of distance
  preserving maps for privacy preserving data mining,'' in \emph{{Proceedings
  of PKDD}}, Berlin, Germany, September 2006.

\bibitem{yiu10}
M.~L. Liu, G.~Ghinita, C.~S.Jensen, and P.~Kalnis, ``Enabling search services
  on outsourced private spatial data,'' \emph{The International Journal of on
  Very Large Data Base}, vol.~19, no.~3, 2010.

\bibitem{rtreebook}
Y.~Manolopoulos, A.~Nanopoulos, A.~Papadopoulos, and Y.~Theodoridis,
  \emph{R-trees: Theory and Applications}.\hskip 1em plus 0.5em minus
  0.4em\relax Springer-Verlag, 2005.

\bibitem{marimont79}
R.~Marimont and M.~Shapiro, ``Nearest neighbour searches and the curse of
  dimensionality,'' \emph{Journal of the Institute of Mathematics and its
  Applications}, vol.~24, pp. 59--70, 1979.

\bibitem{mokbel06}
M.~F. Mokbel, C.~yin Chow, and W.~G. Aref, ``The new casper: Query processing
  for location services without compromising privacy,'' in \emph{{Proceedings
  of Very Large Databases Conference (VLDB)}}, 2006, pp. 763--774.

\bibitem{paillier99}
P.~Paillier, ``Public-key cryptosystems based on composite degree residuosity
  classes,'' in \emph{EUROCRYPT}.\hskip 1em plus 0.5em minus 0.4em\relax
  Springer-Verlag, 1999, pp. 223--238.

\bibitem{papadopoulos10}
S.~Papadopoulos, S.~Bakiras, and D.~Papadias, ``Nearest neighbor search with
  strong location privacy,'' in \emph{{Proceedings of Very Large Databases
  Conference (VLDB)}}, 2010.

\bibitem{franco85}
F.~P. Preparata and M.~I. Shamos, \emph{Computational Geometry: An
  Introduction}.\hskip 1em plus 0.5em minus 0.4em\relax Springer-Verlag, 1985.

\bibitem{rudelson09}
M.~Rudelson and R.~Vershynin, ``Smallest singular value of a random rectangular
  matrix,'' \emph{Communications on Pure and Applied Mathematics}, vol.~62, pp.
  1707--1739, 2009.

\bibitem{shi07}
E.~Shi, J.~Bethencourt, T.-H.~H. Chan, D.~Song, and A.~Perrig,
  ``Multi-dimensional range query over encrypted data,'' in \emph{IEEE
  Symposium on Security and Privacy}, 2007.

\bibitem{sion05}
R.~Sion, ``Query execution assurance for outsourced databases,'' in
  \emph{{Proceedings of Very Large Databases Conference (VLDB)}}, 2005.

\bibitem{wang10icdcs}
C.~Wang, N.~Cao, J.~Li, K.~Ren, and W.~Lou, ``Secure ranked keyword search over
  encrypted cloud data,'' in \emph{{Proceedings of IEEE International
  Conference on Distributed Computing Systems (ICDCS)}}, 2010.

\bibitem{williams08}
P.~Williams, R.~Sion, and B.~Carbunar, ``Building castles out of mud: Practical
  access pattern privacy and correctness on untrusted storage,'' in \emph{{ACM
  Conference on Computer and Communications Security}}, 2008.

\bibitem{wong09}
W.~K. Wong, D.~W.-l. Cheung, B.~Kao, and N.~Mamoulis, ``Secure knn computation
  on encrypted databases,'' in \emph{{Proceedings of ACM SIGMOD
  Conference}}.\hskip 1em plus 0.5em minus 0.4em\relax New York, NY, USA: ACM,
  2009, pp. 139--152.

\bibitem{xie07}
M.~Xie, H.~Wang, J.~Yin, and X.~Meng, ``Integrity auditing of outsourced
  data,'' in \emph{{Proceedings of Very Large Databases Conference (VLDB)}},
  2007, pp. 782--793.

\bibitem{yiu08}
M.~L. Yiu, C.~S. Jensen, X.~Huang, and H.~Lu, ``Spacetwist: Managing the
  trade-offs among location privacy, query performance, and query accuracy in
  mobile services,'' in \emph{{Proceedings of IEEE International Conference on
  Data Engineering (ICDE)}}, Washington, DC, USA, 2008, pp. 366--375.

\end{thebibliography}

\begin{IEEEbiography}[{\includegraphics[width=1in,height=1.25in,clip,keepaspectratio]{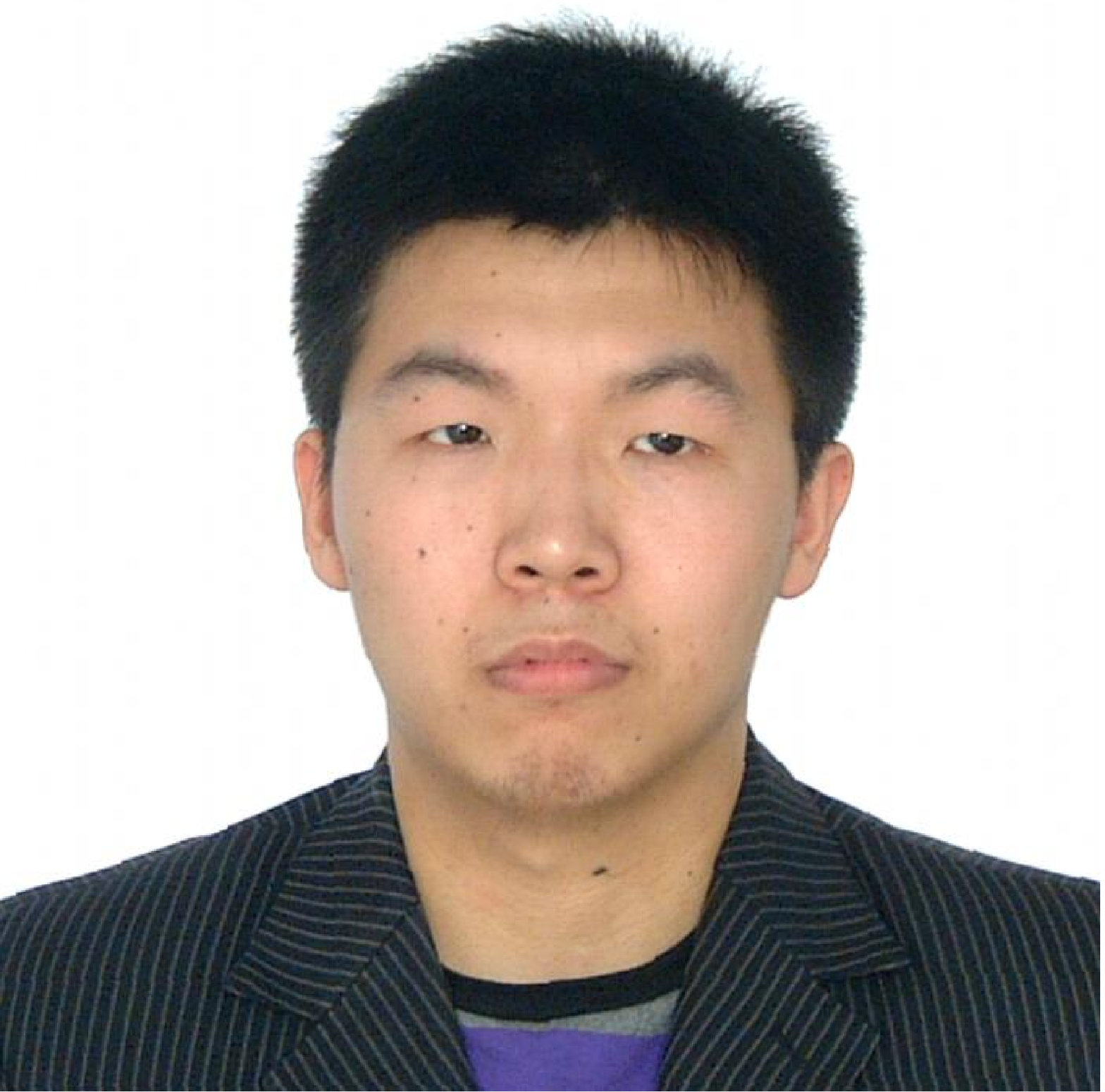}}]{Huiqi Xu} is a PhD student in the Distributed Computing Systems group in the University of Minnesota at Twin Cities. He obtained his Master's degree in Computer Science from Wright State University in June 2012 and his Bachelor's  degree in Computer Science from Chongqing University in June 2009. His research interests include privacy-aware computing and cloud computing.
\end{IEEEbiography}

\begin{IEEEbiography}[{\includegraphics[width=1in,height=1.25in,clip,keepaspectratio]{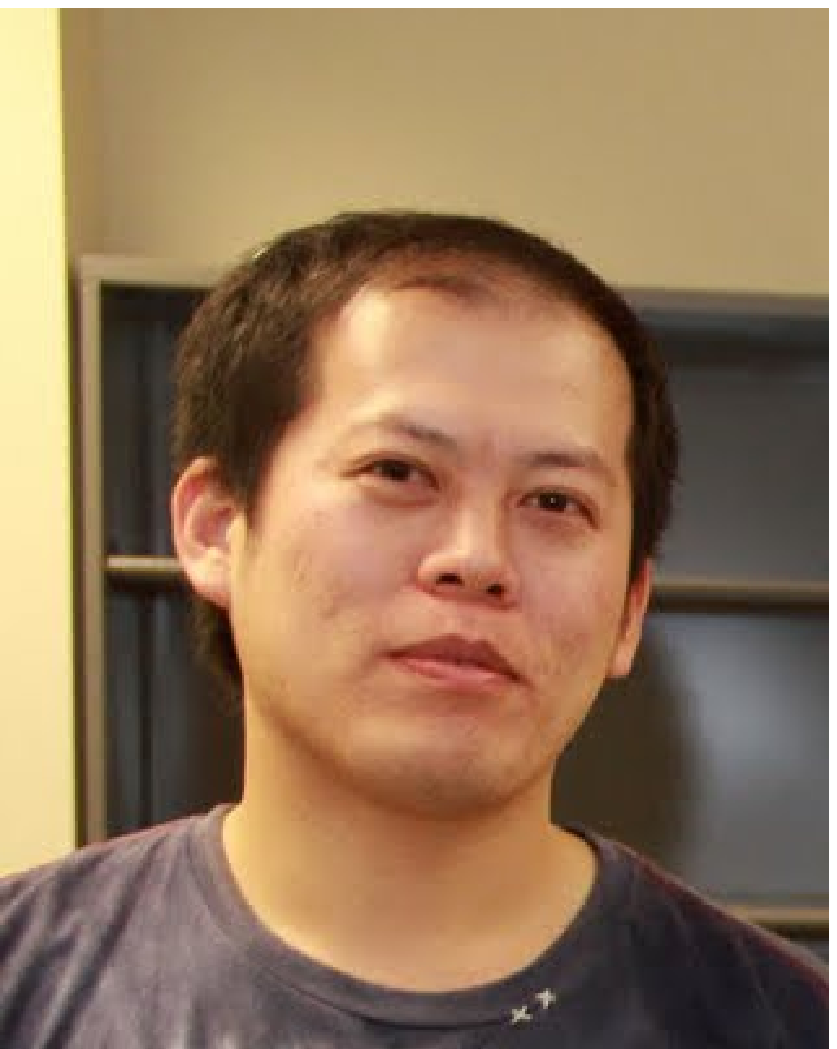}}]{Shumin Guo} is currently a PhD student in the Department of Computer Science and Engineering, and a member of the Data Intensive Analysis and Computing (DIAC) Lab, at Wright State University, Dayton, OH, USA.  He received his Master's degree in Electronics Engineering from Xidian University, Xi'an China, in 2008. His current research interest are privacy preserving data mining, social network analysis and cloud computing.
\end{IEEEbiography}

\begin{IEEEbiography}[{\includegraphics[width=1in,height=1.25in,clip,keepaspectratio]{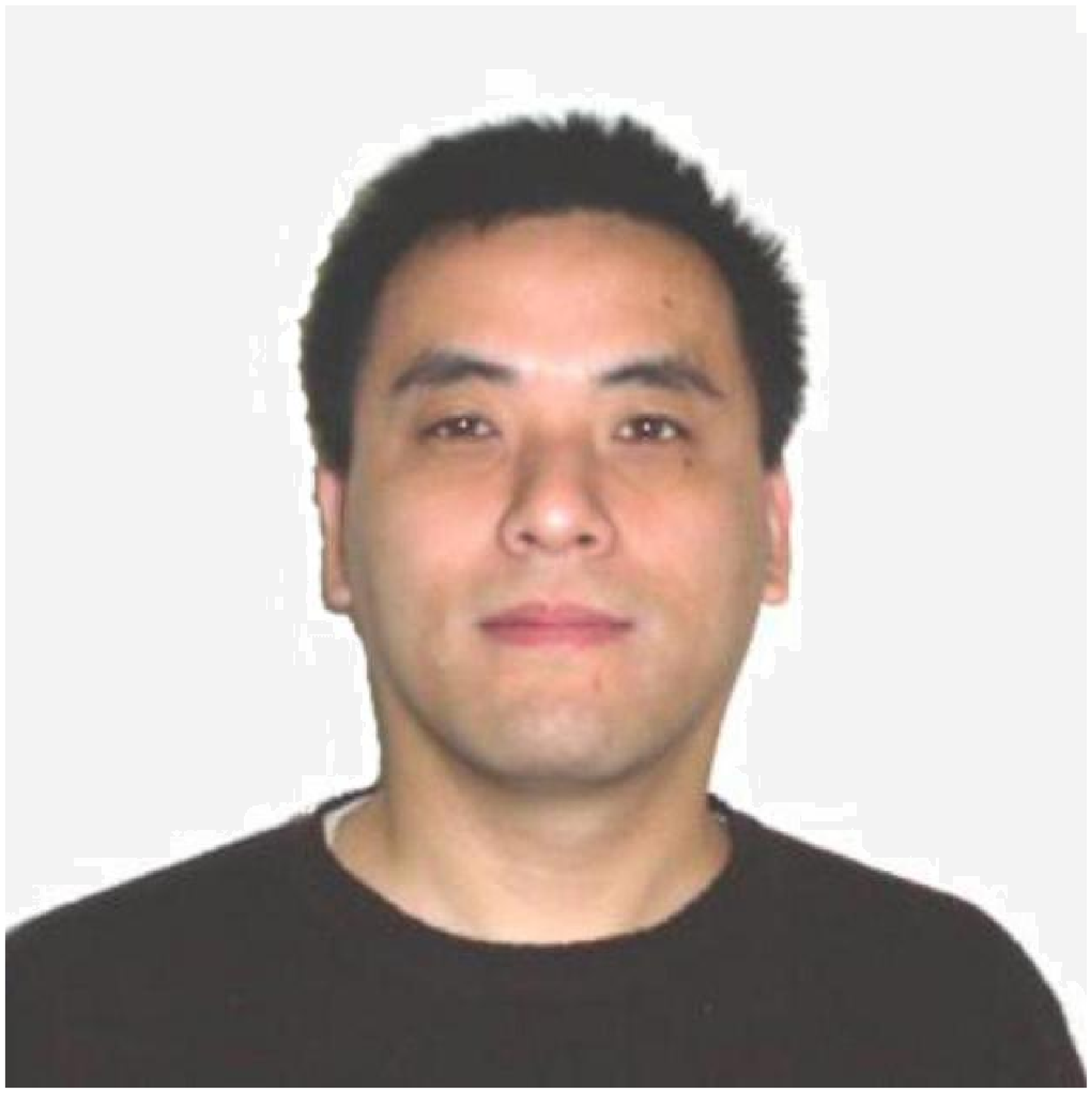}}]
{Keke Chen} is an assistant professor in the Department of Computer Science and Engineering, and a member of the Ohio Center of Excellence in Knowledge Enabled Computing (the Kno.e.sis Center), at Wright State University. He directs the Data Intensive Analysis and Computing (DIAC) Lab at the Kno.e.sis Center. He earned his PhD degree from Georgia Institute of Technology in 2006, his Master's degree from Zhejiang University in China in 1999, and his Bachelor's degree from Tongji University in China in 1996. All degrees are in Computer Science.  His current research areas include visual exploration of big data, secure data services and mining of outsourced data, privacy of social computing, and cloud computing. During 2006-2008, he was a senior research scientist at Yahoo! Labs, working on web search ranking, cross-domain ranking, and web-scale data mining. He owns three patents for his work in Yahoo!.
\end{IEEEbiography}
\newpage
\section{Appendix}
\subsection{Proofs.}
\textbf{Proving that RASP is not OPE.} \\
Let $y=(E_{ope}(x)^T, 1, v)^T$ and we only need to prove that $F(y)=Ay$ does not preserve the dimensional value order. Let $\mathbf{f}^i$ be the selection vector $(0,\dots,1,\dots,0)$ i.e., only the $i$-th dimension is 1 and other dimensions are 0. Then, $(\mathbf{f}^i)^Ty$ will return the value at dimension $i$ of $y$.

\begin{proof}
Let $A$ be an invertible matrix with at least two non-zero entries in each row. For any vector $y$, let $\mathbf{y'}=A\mathbf{y}$. For any two vectors $s$ and $t$, using the dimensional selection vector $f^i$, we have $s_i' = (\mathbf{f}^i)^TA\mathbf{s}$ and $t_i' = (\mathbf{f}^i)^TA\mathbf{t}$ . If the dimensional order is preserved, we will have $(s_i-t_i)(s_i'-t_i') >0$. However,
\begin{eqnarray}
(s_i-t_i)(s_i'-t_i') &=& (s_i-t_i)(\mathbf{f}^i)^TA(s-t)\nonumber \\
&=&(s_i-t_i)\sum_{j=1}^ka_{i,j}(s_j-t_j),
\end{eqnarray}
where $a_{i,j}$ is the $i$-th row $j$-th column element of $A$. Without loss of generality, let's assume $s_i>t_i$ (for $s_i<t_i$ the same proof applies).  It is straightforward to see that the sign of $(s_i-t_i)(s_i'-t_i')$ is subject to the values $s_j$ and $t_j$ in other dimensions $j \neq i$. As a result, RASP does not preserve the dimensional order.
\end{proof}

\textbf{Proving that  $MBR^{(MID)}$ encloses $MBR^{(mid)}$.}
\begin{proof}
In general, the MBR of an arbitrary polyhedron can be derived based on the vertices of the
polyhedron. Based on the property of convexity preserving of RASP, a polyhedron is mapped to another polyhedron in the encrypted space. Concretely, let a polyhedron $P$ has $m$ vertices $\{x_1,\ldots,x_m\}$, which are mapped to the vertices in the encrypted space: $\{y_1,\ldots,y_m\}$. Then, the upper bound and lower bound of dimension  $j$ of the $MBR$ of the polyhedron in the encrypted space are  determined by $\max\{ y_{ij}, i=1\ldots m\}$ and $\min\{y_{ij}, i=1\ldots m\}$, respectively.

Since we only use MBR to reduce the set of results for filtering, a slightly larger MBR would still guarantee the correctness of the MBR based query processing algorithm, with possibly increased filtering cost. In the following, we try to find such a MBR to enclose MBR$^{(mid)}$. By the definition of the square ranges $S^{(low)}$, $S^{(mid)}$ and
$S^{(high)}$, their vertices have the relationship $x_{i}^{(mid)} = (x_{i}^{(low)} +
x_{i}^{(high)})/2$. The images of the vertices are notated as $y_i^{(low)}$, $y_i^{(high)}$, and $y_i^{(mid)}$, respectively. Correspondingly, the MBR$^{(mid)}$ in the perturbed space should be found from $\{y_1^{(mid)},\ldots,y_m^{(mid)}\}$, where $y_i^{(mid)} = A (x_i^{(mid)}, 1, v_i^{(mid)})^T$. Since $(y_{i}^{(low)} + y_{i}^{(high)})/2$ $= A (x_i^{(mid)}, 1, (v_i^{(low)}+v_i^{(high)})/2)^T$, and $(v_i^{(low)}+v_i^{(high)})/2$ is a valid positive random number.
Thus, MBR$^{(mid)}$ can be determined with vertices $\{(y_{i}^{(low)} + y_{i}^{(high)})/2\}$. 

Let the j-th dimension of MBR$^{(L)}$ represented as \\$[s_{j,min}^{(L)}, s_{j,max}^{(L)}]$, where $s_{j,min}^{(L)}$ $= \min\{ y_{ij}^{(L)}, i=1\ldots m\}$, and $s_{j,max}^{(L)} = \max\{ y_{ij}^{(high)}, i=1\ldots m\}$. Now we choose the MBR$^{(MID)}$ as follows: for j-th dimension we use $[(s_{j,min}^{(low)}+s_{j,min}^{(high)})/2, (s_{j,max}^{(low)}+s_{j,max}^{(high)})/2]$. We show that

For two sets of $m$ real values $\{a_1,\ldots, a_m\}$ and $\{b_1,\ldots, b_m\}$, it is easy to verify that
\begin{equation} \label{eq:max}
\max\{a_1,\ldots, a_m\} + \max\{b_1,\ldots, b_m\} \geq \max\{a_1+b_1,\ldots, a_1+b_m\}
\end{equation}
\begin{equation}\label{eq:min}
\min\{a_1,\ldots, a_m\} + \min\{b_1,\ldots, b_m\} \leq \min\{a_1+b_1,\ldots, a_1+b_m\}.
\end{equation}
Thus, $(s_{i,min}^{(low)}+s_{i,min}^{(high)})/2 \leq$ $\min\{(y_{ij}^{(low)}+y_{ij}^{(high)})/2, i=1\ldots m\}$ $= s_{i,min}^{(mid)}$, and $(s_{i,max}^{(low)}+s_{i,max}^{(high)})/2 \geq s_{i,max}^{(mid)}$. Since for each dimension, MBR$^{(MID)}$ encloses MBR$^{(mid)}$, we have MBR$^{(MID)}$ encloses $MBR^{(mid)}$.
\end{proof}

\subsection{Algorithms}
\newcommand{\enc}{\ensuremath{\mbox{\bf RASP\_Perturb}}}
\begin{algorithm}[h]
\small
\caption{RASP Data Perturbation}\label{algo:enc}
\begin{algorithmic}[1]
\State $\enc(X, RNG, RIMG, K_o)$
\State Input: $X$: $k\times n$ data records, $RNG$: random real value generator that draws values from the standard normal distribution, $RIMG$ : random  invertible matrix generator, $K_{ope}$: key for OPE $E_{ope}$; Output: the matrix $A$
\medskip
\State $A \leftarrow 0$;
\State $A_3 \leftarrow$ the last column of $A$;
\State $v_0\leftarrow 4$;
\While {$A_3$ contains zero}
\State generate $A$ with $RIMG$;
\EndWhile
\For{each record $x$ in $X$}
\State $v\leftarrow v_0-1$;
\While {$v<v_0$}
\State $v\leftarrow$ RNG;
\EndWhile
\State $y \leftarrow A((E_{ope}(x, K_{ope}))^T, 1, v)^T$;
\State submit $y$ to the server;
\EndFor
\State return $A$;
\end{algorithmic}
\end{algorithm}

Algorithm \ref{algo:qenc} encodes a normal range query and generate the $Q_i$ matrices and the MBR for the transformed query.
\begin{algorithm}[h]
\small
\caption{RASP Secure Query Transformation.}\label{algo:qenc}
\newcommand{\qt}{\ensuremath{\mbox{\bf QuadraticQuery}}}
\begin{algorithmic}[1]
\State $\qt(Cond, A)$
\State Input: Cond: $2d$ simple conditions for $d$-dimensional data, 2 conditions for each dimension. $A$:the perturbation matrix. Output: the MBR of the transformed range and the quadratic query matrices $Q_i, i=1\ldots 2d$.
\medskip
\State $v_0\leftarrow 4$;
\For{each condition $C_i$ in Cond}
\State $u \leftarrow zeros(d+2,1)$;
\If{$C_i$ is like $X_j < a_j$}
\State $u_{j} \leftarrow 1$, $u_{d+1} \leftarrow -a_j$;
\EndIf
\If{$C_i$ is like $X_j> a_j$}
\State $u_{j} \leftarrow -1$, $u_{d+1} \leftarrow a_j$;
\EndIf
\State $w \leftarrow zeros(d+2,1)$;
\State $w_{d+2} \leftarrow 1$;
\State $w_{d+1} \leftarrow v_0$;
\State $Q_i \leftarrow (A^{-1})^T\mathbf{u}\mathbf{w}^TA^{-1}$;
\EndFor
\State Use the vertex transformation method to find the MBR of the transformed queries;
\State return MBR and $\{Q_i, i=1\ldots 2d\}$;
\end{algorithmic}
\end{algorithm}

In Algorithm \ref{algo:qp}, the two-stage query processing uses the MBR to find the initial query result and then filters the result with the transformed query conditions $y^TQ_iy< 0$, where the matrices \{$Q_i$\}  and the MBR are passed by the client and $y$ is each perturbed record.

\begin{algorithm}[h]
\caption{Two-Stage Query Processing.}\label{algo:qp}
\newcommand{\qp}{\ensuremath{\mbox{\bf ProcessQuery}}}
\begin{algorithmic}[1]
\State $\qp(MBR, \{Q_i\})$
\State Input: MBR: MBR for the transformed query; $\{Q_i\}$:filtering conditions; Output: the set of perturbed records satisfying the conditions.
\medskip
\State $Y \leftarrow $ use the indexing tree to find answers for MBR;
\State $Y'\leftarrow \emptyset$;
\For{each record $y$ in Y}
\State success $\leftarrow 1$
\For{each condition $Q_i$}
\If{$y^TQ_iy \geq 0$}
\State    success $\leftarrow 0$;
\State break;
\EndIf
\EndFor
\If{success = 1}
\State add $y_i$ into $Y'$;
\EndIf
\EndFor
\State return $Y'$ to the client;
\end{algorithmic}
\end{algorithm}

The following Algorithm \ref{alg:k-delta} describes the details of the $(K, \delta)$-Range algorithm for determining the inner range.
\begin{algorithm}[h]
\caption{$(K, \delta)$-Range Algorithm}\label{alg:k-delta}
\begin{algorithmic}[1]
\Procedure{$(K, \delta)$-Range}{$L_{1}, L_{m}, k, \delta$}
\State $high \gets L_{m}$, $low \gets L_{1}$;
\While{$high - low \geq \mathcal{E}$}
\State $mid \gets (high + low)/2 $;
\State num $\gets$ number of points in $S^{(mid)}$;
\If {$ num \geq k \&\& num \leqslant k+\delta$}
\State Break the loop;
\ElsIf {$num > k+delta$}
\State $ high \gets mid $;
\Else
\State $ low \gets mid $;
\EndIf
\EndWhile
\State \textbf{return} $S^{(mid)}$;
\EndProcedure
\end{algorithmic}
\end{algorithm}

\end{document}